\documentclass[12pt]{article}
\usepackage{amsmath,amssymb,epsf,epsfig,amsthm,times,ifthen,epic,psfrag}
\usepackage{enumerate}
\usepackage{printtime}
\usepackage{ifthen}
\usepackage{cite}     
\usepackage{fancyheadings}
\usepackage{lastpage} 

\def\submitteddate{November 19, 2013}

\newtheorem{thm}              {Theorem}     [section]
\newtheorem{lem}      [thm] {Lemma}
\newtheorem{cor}  [thm] {Corollary}

\theoremstyle{definition}
\newtheorem{defn} [thm] {Definition}

\begin{document}

\newcommand{\creationtime}{\today \ \ @ \theampmtime}

\pagestyle{fancy}
\renewcommand{\headrulewidth}{0cm}
\chead{\footnotesize{Dougherty-Freiling-Zeger}}
\rhead{\footnotesize{\submitteddate}}
\lhead{}
\cfoot{Page \arabic{page} of \pageref{LastPage}} 

\renewcommand{\qedsymbol}{$\blacksquare$} 


\newtheorem{theorem}              {Theorem}    
\newtheorem{corollary}  [theorem] {Corollary}
\newtheorem{proposition}[theorem] {Proposition}
\newtheorem{remark}     [theorem] {Remark}
\newtheorem{algorithm}  [theorem] {Algorithm}
\newtheorem{conjecture} [theorem] {Conjecture}
\newtheorem{example}    [theorem] {Example}

\theoremstyle{definition}
\newtheorem{definition} [theorem] {Definition}
\newtheorem*{claim}  {Claim}
\newtheorem*{notation}  {Notation}
\newcommand{\Dim}{\mbox{dim}}
\newcommand{\Codim}{\mbox{codim}}
\newcommand{\Comment}[1]{& [\mbox{from  #1}]}
\renewcommand{\emptyset}{\varnothing} 
\renewcommand{\subset}{\subseteq}     
\newcommand{\TBA}{*** To Be Added ***}
\newcommand{\Compose}{} 

\begin{titlepage}

\setcounter{page}{0}

\title{Characteristic-Dependent Linear Rank Inequalities with Applications to Network Coding
\thanks{This work was supported by the Institute for Defense Analyses and the
National Science Foundation.\newline
\indent \textbf{R. Dougherty} is with the 
Center for Communications Research,
4320 Westerra Court, 
San Diego, CA 92121-1969 (rdough@ccrwest.org).\newline
\indent \textbf{E. Freiling and K. Zeger} are with the 
Department of Electrical and Computer Engineering, 
University of California, San Diego, 
La Jolla, CA 92093-0407
(efreilin@ucsd.edu, zeger@ucsd.edu).
}}

\author{Randall Dougherty, Eric Freiling, and Kenneth Zeger\\ \ }

\date{
\textit{
IEEE Transactions on Information Theory\\
Submitted: \submitteddate
}}

\maketitle
\begin{abstract}
Two characteristic-dependent linear rank inequalities are given for eight variables.
Specifically, the first inequality holds for all finite fields whose characteristic
is not three and does not in general hold over characteristic three.
The second inequality holds for all finite fields whose characteristic
is three and does not in general hold over characteristics other than three.
Applications of these inequalities to the computation of 
capacity upper bounds in network coding are demonstrated.
\end{abstract}

\thispagestyle{empty}
\end{titlepage}

\newpage

\section{Introduction}

The study of information inequalities is a subfield of information theory that
describes linear constraints on the entropies of 
finite collections of jointly distributed discrete random variables.
Historically, the known information inequalities were orignally all special
cases of Shannon's conditional mutual information inequality $I(X;Y|Z)\ge 0$,
but later were generalized to other types of inequalities, called non-Shannon inequalities.
Information inequalities have been shown to be useful for computing upper bounds
on the network coding capacities of certain networks.

Analagously,
the study of linear rank inequalities is a topic of linear algebra,
which describes linear constraints on the dimensions of collections of subspaces of finite dimensional vector spaces.
In fact, the set of all information inequalities can be viewed as subclass of
the set of all linear rank inequalities.

Information inequalities hold over all collections of a certain number of random variables.
In constrast, linear rank inequalities may hold over only certain vector spaces,
such as those whose scalars have particular field characteristics.

In this paper,
we present two new linear rank inequalities over finite fields,
which are not information inequalities,
and with the peculiar property
that they only hold for certain fields, 
depending on the associated vector space.
The first inequality is shown to hold over all vector spaces when the field characteristic is anything but three
(Theorem~\ref{thm:T8}),
but does not always hold when the field characteristic is three
(Theorem~\ref{thm:T8-char3}).
In contrast,
the second inequality is shown to hold over all vector spaces when the field characteristic is three
(Theorem~\ref{thm:nonT8}),
but does not always hold when the field characteristic is not three
(Theorem~\ref{thm:non-T8-non-char3}).
We also show how these inequalities can be used to obtain bounds on the capacities of
certain networks
(Corollaries~\ref{cor:T8-capacity} and \ref{cor:nonT8-capacity}).

It will be assumed that the reader has familiarity with linear algebra,
finite fields, information theory, and network coding.
Nevertheless, we will give some brief tutorial descriptions of these topics for completeness.

\subsection{Background}

In 2000, Ahlswede, Cai, Li, and Yeung introduced the field of Network Coding~\cite{networkcoding1} 
and showed that coding can outperform routing in 
directed acyclic networks.%
\footnote{In what follows, by ``network'' we shall always mean a directed acyclic network.}
There are presently no known 
algorithms to determine the capacity or the linear capacity of a given network.
In fact, it is not even known if such algorithms exist.

Information inequalities are linear inequalities that hold for all jointly
distributed random variables,
and Shannon inequalities are information
inequalities of a certain form~\cite{shannon}. 
Both are defined in Section~\ref{section:LRI}. 
It is known~\cite{Yeung-book} that all information
inequalities containing three or fewer variables are Shannon inequalities. 
The first ``non-Shannon'' information inequality was of four variables and was
published in 1998 by Zhang and Yeung~\cite{Zhang-Yeung}. 
Since 1998, various other non-Shannon inequalities have been found,
for example, by
Ln\v{e}ni\v{c}ka~\cite{Lnenicka}, 
Makarychev, Makarychev, Romashchenko, and Vereshchagin~\cite{MMRV}, 
Zhang~\cite{Zhang13}, 
Zhang and Yeung~\cite{Zhang-Yeung14}, 
Dougherty, Freiling, and Zeger~\cite{DFZ2}, and
Mat\'{u}\v{s}~\cite{Matus}. 
Additionally, in 2007, Mat\'{u}\v{s} demonstrated an infinite collection of independent
non-Shannon information inequalities~\cite{Matus}
and there were necessarily an infinite number of such inequalities.
In 2008, Xu, Wang, and Sun~\cite{XuWangSun} also gave an infinite list of inequalities
but did not establish their necessity.

There is a close connection between information inequalities and network coding~\cite{ChanGrant}. 
Capacities of some networks have been computed by finding
matching lower and upper bounds~\cite{nonshannon}. 
Lower bounds have been found by deriving coding solutions. 
Upper bounds have been found by using
information inequalities and treating the sources as independent random variables 
that are uniformly distributed over the alphabet.
One ``holy grail'' problem of network coding is to develop an algorithm to compute the coding capacity of an arbitrary network.
If such an algorithm exists, 
information inequalities may potentially play a role in the solution.

It has been shown that linear codes are insufficient for network coding in
general~\cite{insufficient}. 
However, linear codes may be desirable to use in practice due to ease of analysis
and implementation.
It has been shown that the coding capacity
is independent of the alphabet size~\cite{Cannons-Dougherty-Freiling-Zeger05}.
However, the linear coding capacity is dependent on alphabet size, 
or more specifically the field characteristic.
In other words, one can potentially achieve a higher rate of linear
communication by choosing one characteristic over another.
To provide 
upper bounds for the linear coding capacity for a particular field one can look
at linear rank inequalities~\cite{rateregions-journal}.
Linear rank inequalities are
linear inequalities that are always satisfied by ranks%
\footnote{Throughout this paper, 
we will use the terminology ``rank'' of a subspace to mean the dimension of the subspace
(i.e. the rank of a matrix whose columns are a basis for the subspace),
in order to parallel the terminology of matroid theory.
}
of subspaces of a vector space.
All information inequalities are linear rank inequalities but not all
linear rank inequalities are information inequalities.
The first example of a
linear rank inequality that is not an information inequality was found by
Ingleton~\cite{ingleton}.
Information inequalities can provide an upper bound
for the capacity of a network, but this upper bound would hold for all alphabets.
Therefore, to determine the linear coding capacity over a certain
characteristic one would have to consider linear rank inequalities.

All linear rank inequalities up to and including five variables are known
and none of these depend on the vector spaces' field characteristics~\cite{5var}.
The set of all linear rank inequalities for six variables has not yet been determined.
Characteristic-dependent linear rank inequalities 
are given, for example, in
~\cite{Blasiak-Kleinberg-Lubetzky}
and ~\cite{rateregions-journal}.

An inequality is given in~\cite{rateregions-journal} which is valid for characteristic two and another inequality is given which is valid
for every characteristic except for two.
These inequalities are then used to
provide upper bounds for the linear coding capacity of two networks.

In the present paper, we give two characteristic-dependent linear rank inequalities on eight variables.
One is valid for characteristic three and the other is valid for
every characteristic except for three.
These inequalities are then used to
provide upper bounds for the linear coding capacity of two networks.

It is our intention that the techniques presented here may prove useful or
otherwise motivate further progress in determining network capacities.

\subsection{Matroids}

In this section a very brief review of matroids is given which
will enable discussion in subsequent sections of a matroid-based
method for constructing a particular network that helps in the
derivation of the linear rank inequalities presented in this paper.

A matroid is an abstract structure that captures a notion of ``independence" that is found in 
finite dimensional vector spaces, graphs, and various other mathematical topics.
We will follow the notation and results of~\cite{matroidBook}.

\begin{defn}
A \textit{matroid}, $M$, is a pair $(E,I)$, where $E$ is a
finite set and $I$ is a set of subsets of $E$ that satisfies the
following properties:
\begin{enumerate}
    \item[(I1)] $\emptyset \in I$.
    \item[(I2)] $\forall A,B\subset E$, if $A \subseteq B \in I$, then $A \in I$.
    \item[(I3)] $\forall A,B\subset E$, 
           if $A,B \in I$ and $|A| > |B|$, 
           then $\exists u \in A\setminus B$ such that $B \cup \{u\} \in I$.
\end{enumerate}
\end{defn}
The sets in $I$ are called \textit{independent sets}.
If a subset of $E$ is not in $I$, then it is called \textit{dependent}.

An example of a matroid is obtained from linear algebra.
Let $F$ be a finite field and let $V(m,F)$ be the vector space of all $m$-dimensional vectors whose
components are elements of $F$.
Suppose $A$ is an $m\times n$ matrix over $F$.
Let $E = \{1,\ldots,n\}$ and $I$ be the set of all $X \subseteq E$ such
that the multiset of columns of $A$ indexed by the elements of $X$ is linearly
independent in the vector space $V(m,F)$.
Then $M = (E,I)$ is a matroid called the \textit{vector matroid} of $A$.

A matroid is said to be \textit{representable} over the field $F$
if it is isomorphic to some vector matroid over $V(m,F)$.

For example, if $F$ is the binary field and
$$A = \bordermatrix{& a & b & c & d & e\cr
                    & 1 & 0 & 0 & 1 & 1\cr
                    & 0 & 1 & 0 & 0 & 1}$$
where $a,b,c,d,e$ denote the columns of $A$ from left to right,
then $M = (E,I)$ is a vector matroid of $A$, 
where $E = \{a,b,c,d,e\}$ and
\begin{align*}
I &=  \{\emptyset, \{a\}, \{b\}, \{d\}, \{e\}, \{a,b\}, \{a,e\}, \{b,d\}, \{b,e\}, \{d,e\}\}.
\end{align*}

A \textit{base} is a maximal independent set.
Let $B(M)$ denote the set of all bases of a matroid $M$.
In our example, 
$$B(M) = \{\{a,b\}, \{a,e\}, \{b,d\}, \{b,e\}, \{d,e\}\}.$$
It is well known that all the bases of a matroid are of the same cardinality.

If we let $X \subseteq E$ and $I|X = \{i \subseteq X: i \in I\}$, 
then it is easy to see that $(X,I|X)$ is a matroid.
The \textit{rank} of $X$, denoted by $r(X)$, is defined to be the cardinality of a base in $M|X$.
In our example, $r(M) = 2$.
A \textit{circuit} is a minimal dependent set.
The circuits in our example are $\{\{c\}, \{a,d\}, \{a,b,e\}, \{b,d,e\}\}$.

\subsection{Information Theory and Linear rank Inequalities}\label{section:LRI}

In this section we will use the information theoretic concepts of entropy and mutual information
to define and use the linear algebraic concept of linear rank inequalities.
Connections between information inequalities and linear rank inequalities is also discussed.

Let $A,B,C$ be collections of discrete random variables over a finite alphabet
$\mathcal{X}$, and let $p$ be the probability mass function of $A$.
The
\textit{entropy} of $A$ is defined by
\begin{align*}
    H(A) &=  -\sum_u p(u)\log_{|\mathcal{X}|}p(u).
\end{align*}
The \textit{conditional entropy} of $A$ given $B$ is
\begin{align}
    H(A|B) &=  H(A,B) - H(B)\label{h1},
\end{align}
the \textit{mutual information} between $A$ and $B$ is
\begin{align}
    I(A;B) &=  H(A) - H(A|B) = H(A) + H(B) - H(A,B),\label{h2}
\end{align}
and the \textit{conditional mutual information} between $A$ and $B$ given $C$ is
\begin{align}
    I(A;B|C) &= H(A|C) - H(A|B,C) = H(A,C) + H(B,C) - H(C) - H(A,B,C)\label{h3}.
\end{align}
We will make use of the following basic information-theoretic facts \cite{Yeung-book}:
\begin{align}
    0 &= H(\emptyset) \label{h4}\\
    0 &\le H(A) = H(A|\emptyset)\label{h5}\\
    0 &\le H(A|B)\label{h6}\\
    0 &\le I(A;B)\label{h7}\\
    H(A,B|C) &\le H(A|C) + H(B|C)\label{h8}\\
    H(A|B,C) &\le H(A|B) \leq H(A,C|B)\label{h9}\\
    I(A;B,C) &= I(A;B|C) + I(A;C). \label{h12}
\end{align}
The equations (\ref{h5})-(\ref{h9}) were originally given by Shannon in 1948
\cite{shannon}, and can all be obtained from the single inequality $I(A;B|C) \geq 0$.
\begin{defn}
Let $q$ be a positive integer, and let $S_1,\ldots,S_k$ be subsets of
$\{1,\ldots,q\}$.
Let $\alpha_i \in \mathbb{R}$ for $1\leq i \leq k$.
A linear
inequality of the form
\begin{align}
    \alpha_1 H(\{A_i:i\in S_1\}) + \dots + \alpha_k H(\{A_i:i\in S_k\}) &\ge 0
   \label{inequality-format}
\end{align}
is called an \textit{information inequality} if it holds for all jointly distributed random variables $A_1,\ldots,A_q$.
\end{defn}
As an example, taking $q = 2$, $S_1 = \{1\}$, $S_2 = \{2\}$, $S_3 = \emptyset$,
$S_4 = \{1,2\}$, $\alpha_1 = \alpha_2 = 1$, $\alpha_4 = -1$, and using
(\ref{h8}) shows that $H(A_1) + H(A_2) - H(A_1,A_2) \ge 0$ is an information
inequality.

\bigskip
A \textit{Shannon information inequality} is any information inequality that can be expressed as a finite sum of the form
\begin{align*}
& \sum_i \alpha_i I(A_i;B_i|C_i) \ge 0
\end{align*}
where each $\alpha_i$ is a nonnegative real number.
Any information inequality that cannot
be expressed in the form above will be called a \textit{non-Shannon information inequality}.

Linear rank inequalities are closely
related to information inequalities.
In fact, in order to describe linear rank inequalities
we will borrow notation from information theory to use in the context of
linear algebra in the following manner.

Suppose $A$ and $B$ are subspaces of a given vector space $V$,
and let $\langle A, B\rangle$ denote the span of $A\cup B$.
We will let $H(A)$ denote the rank of $A$, 
and let $H(A,B)$ denote the rank of $\langle A, B\rangle$.
The meanings of some other information theoretic notation in the context of linear algebra
then follows from~\eqref{h1}-\eqref{h3}.
Specifically, note that
the conditional entropy notation $H(A|B)$ 
denotes the excess rank of subspace $A$ over that of subspace $A\cap B$, 
or equivalently, the codimension of $A\cap B$ in $A$;
and the mutual information notation $I(A;B)$ denotes the rank of $A\cap B$.

A \textit{linear rank inequality} 
over a vector space $V$
is a linear inequality 
of the form in \eqref{inequality-format},
that is satisfied
by every assignment of subspaces of $V$ to the variables $A_1, \dots, A_q$.

All information inequalities are
linear rank inequalities over all finite vector spaces, 
but not all linear rank inequalities are information inequalities.
For background material on these concepts, the reader is referred to
Hammer, Romashchenko, Shen, and Vereshchagin \cite{rank}.

The first known example of a linear rank inequality 
over all finite vector spaces
that is not an information inequality is the \textit{Ingleton inequality}~\cite{ingleton}:
\begin{align*}
I(A;B) &\le  I(A;B|C) + I(A;B|D) + I(C;D).
\end{align*}
To see that the Ingleton inequality is not an information inequality,
let $A,B,C,D$ be binary random variables, and let $X = (A,B,C,D)$ with probabilities:
\begin{align*}
    P(X=0000) &= 1/4\\
    P(X=1111) &= 1/4\\
    P(X=0101) &= 1/4\\
    P(X=0110) &= 1/4.
\end{align*}
Then the Ingleton inequality fails since:
\begin{align*}
\underbrace{I(A;B)}_{(5-\log_2 27)/2}
&- \underbrace{I(A;B|C)}_{0} - \underbrace{I(A;B|D)}_{0} - \underbrace{I(C;D)}_{0}
> 0.
\end{align*}

\subsection{Network Coding}

In this section,
we will briefly
review some concepts of network coding.
This will enable the discussion later in this paper of our construction of linear rank inequalities
using networks constructed from two particular matroids (T8 and non-T8).
For more details on network coding, see \cite{networkcodingbook}.

A \textit{network} is a finite, directed, acyclic multigraph with messages and demands.
Network \textit{messages} are arbitrary vectors of $k$ symbols over a finite alphabet $\mathcal{A}$.
Each network edge carries a vector of $n$ symbols from $\mathcal{A}$.
Each message originates at
a particular node called the \textit{source node} for that message and is
required by one or more \textit{demand nodes}.
When we draw a network, 
a message variable appearing above a node indicates the message is generated by such node%
\footnote{
We note that in Figures~\ref{fig:T8} and \ref{fig:nonT8},
for convenience, 
we label source messages above nodes lying in both the top and bottom layers in each diagram.
This is meant to indicate that there is, in fact, a separate (but hidden) distinct node for each such
source message,
whose out-edges go directly to the nodes labeled by the source message in the top and bottem layers.
},
and 
a message variable appearing below a node indicates the message is demanded by such node,
For a given network, the values of $k$ and $n$ can be chosen in order to implement certain
codes and to obtain certain throughput $k/n$.

The inputs to a network node are the vectors carried on its in-edges as well as
the messages, if any, generated at the node.
The outputs of a network node are the
packets carried on its out-edges as well as any demanded messages at the node.
Each output of a node must be a function only of its inputs.
A \textit{coding solution} for the
network is an assignment of such functions to the network edges.
When the values of $k$ and $n$ need
to be emphasized, the coding solution will be called a $(k,n)$-coding solution.
The \textit{capacity} of a network is defined as:
\begin{align*}
    \mathcal{C} &=  \sup\{ k/n : \exists\mbox{ a $(k,n)$-coding solution}\}.
\end{align*}

A solution is called a \textit{linear solution},
if the alphabet $\mathcal{A}$ is a finite field
and the edge functions are linear
(i.e.  linear combinations of their input vectors where the coefficients are matrices over the field).

The \textit{linear capacity} is defined the same
as the capacity but restricting solutions to be linear.
It is also easily verified that if $x$ is a message, 
then $H(x) = k$, 
and if $x$ is a vector carried by an edge, then $H(x) \leq n$.

\begin{figure}
\begin{center}
\psfrag{n1}{\LARGE $n_{1}$}
\psfrag{n2}{\LARGE $n_{2}$}
\psfrag{n3}{\LARGE $n_{3}$}
\psfrag{n4}{\LARGE $n_{4}$}
\psfrag{n5}{\LARGE $n_{5}$}
\psfrag{n6}{\LARGE $n_{6}$}
\includegraphics[scale=.75]{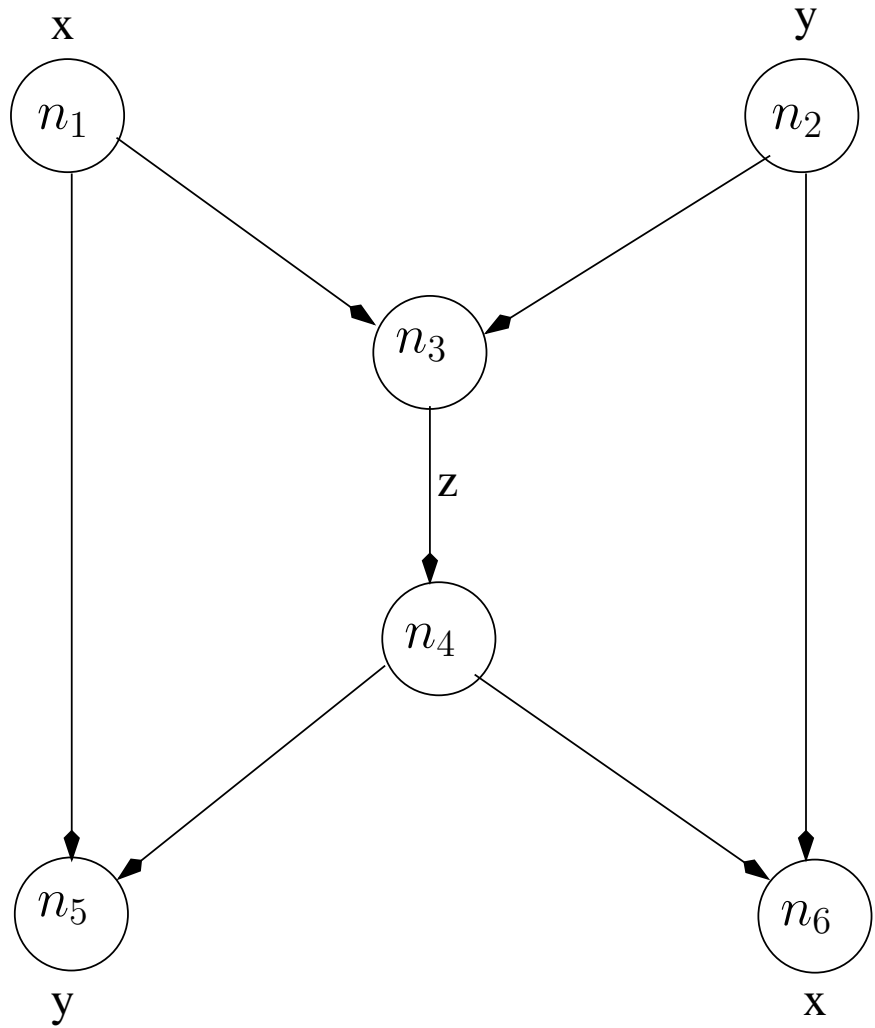}
\end{center}
\caption{\label{fig:butterfly} The Butterfly network with source messages $x$ and $y$,
generated by source nodes $n_1$ and $n_2$, respectively. 
Demand nodes $n_5$ and $n_6$ demand messages $y$ and $x$, respectively.
}
\end{figure}

Let us illustrate a method for finding capacity bounds by examining the well-known Butterfly network,
depicted in Figure~\ref{fig:butterfly}.
We assume the network messages $x$ and $y$
are independent, $k$-dimensional, random vectors with uniformly distributed components.
Then in any solution it must be the case that
\begin{align}
    H(y|x,z) &=  0 \label{butterfly:1}
\end{align}
since $y$ is a function of $x$ and $z$,
and also that
\begin{align}
    2k &= H(x) + H(y)\\
        &= H(x,y) & \text{[from indep. of x and y]}\nonumber\\
        &\leq H(x,y,z) & \text{ [from \eqref{h9}]}\nonumber\\
        &= H(x,z) + H(y|x,z) & \text{[from \eqref{h1}]}\nonumber\\
        &= H(x,z) & \text{[from (\ref{butterfly:1})]}\nonumber\\
        &\le H(x) + H(z) & \text{[from \eqref{h8}]}\nonumber\\
        &\le k + n. \label{rv-inequality}
\end{align}
This implies $2k \le k + n$,
or equivalently $k/n \le 1$.
Since this bound holds for all choices of $k$ and $n$,
the coding capacity must be at most $1$.
On the other hand,
a solution with $k = n = 1$ is obtained by taking
$z = x+y$ over any finite field alphabet,
so the coding capacity is at least $1$.
Thus the coding capacity for the Butterfly network is the same as the linear coding
capacity which is exactly equal to $1$.

The inequalities in \eqref{rv-inequality} were based on random variables $x,y,z$.
Later, in the proofs of Corollaries~\ref{cor:T8-capacity} and~\ref{cor:nonT8-capacity},
we will obtain bounds on the capacities of networks by using linear rank inequalities,
instead of information inequalities.
In those cases, certain vector subspaces will be used instead of random variables,
but the procedure will appear similar.

\newpage

\section{Preliminaries}

In this section, we given some technical lemmas which will be useful for proving the main results of the paper.

If $A$ is a subspace of vector space $V$,
and $\overline{A}$ is a subspace of $A$,
then we will use the notation 
$\Codim_A(\overline{A}) = \Dim(A) - \Dim(\overline{A})$
to represent the codimension 
of $\overline{A}$ in $A$.
We will omit the subscript when
it is obvious from the context which space the codimension is with respect to.

\begin{lem}\label{lemma1}
\cite{rateregions-journal} Let $V$ be a finite dimensional vector space with subspaces
$A$ and $B$.
Then the subspace $A \cap B$ has codimension at most
$\Codim(A) + \Codim(B)$ in $V$.
\end{lem}

\begin{proof}
We know $H(A) + H(B) - I(A; B) = H( A,B) \le H(V)$.
Then adding $H(V)$ to both
sides of the inequality gives $H(V) - I(A; B) \le H(V)-H(A) + H(V) - H(B)$.
Thus, $\Codim(A\cap B) \le \Codim(A) + \Codim(B)$.
\end{proof}

\begin{lem}\label{lemma2}
\cite{rateregions-journal} Let $A$ and $B$ be vector spaces over the same finite scalar field and with subspaces
$\overline{A}$ and $\overline{B}$, respectively.
Let $f:A \rightarrow B$ be a
linear function such that $f(A\backslash \overline{A}) \subseteq B\backslash \overline{B}$.
Then the codimension of $\overline{A}$ in $A$ is at most the
codimension of $\overline{B}$ in $B$.
\end{lem}

\begin{proof}
Suppose 
a base for $A$ consists of a base for $\overline{A}$
together with the vectors $a_1, \ldots, a_n$.
Let $\gamma_1,\ldots,\gamma_n$ be field elements which are not all zero.
Then $\gamma_1 a_1+ \cdots + \gamma_n a_n\in A\setminus\overline{A}$,
so 
$\gamma_1f(a_1)+ \cdots + \gamma_nf(a_n) = f(\gamma_1 a_1+ \cdots + \gamma_n a_n) \in B\setminus\overline{B}$.
Thus, the vectors
$f(a_1),\ldots,f(a_n)$ are linearly independent over the subspace $\overline{B}$,
and therefore 
$\Codim_A(\overline{A}) 
= n \le   
\Codim_B(\overline{B})$.
\end{proof}

\begin{lem}\label{lemma3}
\cite{rateregions-journal} 
Let $A$ and $B$ be vector spaces over the same finite scalar field,
let $\overline{B}$ be a subspace of $B$, 
and let $f:A \rightarrow B$ be a linear function.
Then $f(t) \in \overline{B}$ on a subspace of $A$ of codimension at most the codimension of
$\overline{B}$.
\end{lem}

\begin{proof}
Let $\overline{A} = \{t \in A: f(t) \in \overline{B}\}$.
Then $f(A\backslash \overline{A}) \subseteq B\backslash \overline{B}$
and the result follows from Lemma~\ref{lemma2}.
\end{proof}

\begin{lem}\label{lemma4}
\cite{rateregions-journal} Let $V$ be a finite dimensional vector space and let
$A_1,\ldots,A_k,B$ be subspaces of $V$.
Then for $i = 1,\dots,k$, 
there exist linear functions $f_i:B \rightarrow A_i$ such that $f_1 + \cdots + f_k = I$ on
a subspace of $B$ of codimension $H(B|A_1,\ldots,A_k)$.
\end{lem}

\begin{proof}
Let $W$ be a subspace of $B$ defined by $W = \langle A_1, \ldots, A_k\rangle \cap B$.
The subspace on which this lemma holds is $W$.
If $H(W) = 0$ , then
the lemma would be trivially true.
So, assume that $H(W) > 0$, and let
$\{w_1,\ldots,w_n\}$ be a basis for $W$.
For each $j = 1,\ldots, n$, choose
$x_{i,j} \in A_i$ for $i = 1,\ldots, k$ such that $w_j = x_{1,j} + \dots + x_{k,j}$.
For each $i = 1,\ldots,k$, define a linear mapping $g_i : W \rightarrow A_i$ so that $g_i(w_j) = x_{i,j}$ for all $i$ and $j$.
Then extend
$g_i$ arbitrarily to $f_i:B\rightarrow A_i$.
Now we have linear functions
$f_1,\ldots,f_k$ such that $f_1 + \cdots + f_k = I$ on $W$.
The dimension of
$W$ is $H(W) = I(A_1,\ldots,A_k;B)$, so the codimension of $W$ is $H(B) - I(A_1,\ldots,A_k;B) = H(B|A_1,\ldots,A_k)$.
\end{proof}

\begin{lem}\label{lemma5}
\cite{rateregions-journal} Let $V$ be a finite-dimensional vector space and let $A,B$,
and $C$ be subspaces of $V$.
Let $f:A\rightarrow B$ and $g:A \rightarrow C$ be
linear functions such that $f+g = 0$ on $A$.
Then $f = g = 0$ on a subspace of
$A$ of codimension at most $I(B;C)$.

\begin{proof}
Let $K$ be the kernel of $f$.
Clearly, $f$ maps $A$ into $B\cap C$ and since $f$ is linear 
the rank of its domain is at most the sum of the ranks of its kernel and range, so
\begin{align*}
    \Codim(K) &=H(A) - H(K) \le I(B;C).
\end{align*}
\end{proof}
\end{lem}

\begin{lem}\label{lemma6}
\cite{rateregions-journal} Let $V$ be a finite dimensional vector space and let
$A,B_1,\dots,B_k$ be subspaces of $V$.
For each $i = 1,\ldots,k$ let $f_i:A \rightarrow B_i$ be a linear function such that $f_1 + \cdots + f_k = 0$
on $A$.
Then $f_1 = \cdots = f_k = 0$ on a subspace of $A$ of
codimension at most $H(B_1) + \cdots + H(B_k) - H(B_1,\ldots,B_k)$.
\end{lem}

\begin{proof}
First we apply Lemma~\ref{lemma5} to $f_1$ and $(f_2 + \dots + f_k)$ to get
$f_1 = (f_2 + \dots + f_k) = 0$ on a subspace $A_1$ of $A$ of codimension at
most $I(B_1; B_2,\ldots,B_k) = H(B_1) + H(B_2,\ldots,B_k) - H(B_1,B_2,\ldots,B_k)$.
Then apply Lemma~\ref{lemma5} to $f_2$ and $(f_3 + \dots + f_k)$ to get $f_2 = (f_3 + \dots + f_k) = 0$ on a subspace $A_2$ of
$A_1$ of codimension at most $I(B_2; B_3,\ldots,B_k) = H(B_2) + H(B_3,\ldots,B_k) - H(B_2,B_3,\ldots,B_k)$.
Continue on until we apply Lemma
\ref{lemma5} to $f_{k-1}$ and $f_k$ to get $f_{k-1} = f_k = 0$ on a subspace
$A_{k-1}$ of $A_{k-2}$ of codimension at most $I(B_{k-1}; B_k) = H(B_{k-1}) + H(B_k) - H(B_{k-1},B_k)$.
Now $A_{k-1}$ is a subspace of $A$ of codimension at
most $H(B_1) + \dots + H(B_k) - H(B_1,\ldots,B_k)$, on which $f_1 = f_2 = \dots
= f_k = 0$.
\end{proof}

\begin{lem}\label{lem Inj}
Let $A,B,C,D,E$ be subspaces of a vector space $V$ and let $f_R,f_L,g_R$,
and $g_L$ be functions such that $f_R:A \rightarrow C, f_L:C \rightarrow A, g_R:B \rightarrow D,$ and $g_L:D \rightarrow E$.
If $f_Lf_R = I$ on $A$ and
$g_Lg_R$ is injective on $B$, then $g_Lf_R$ is injective on $f_L(f_RA \cap g_RB)$.
\end{lem}

\begin{proof}
Let $x,y \in f_L(f_RA \cap g_RB)$.
We know $f_R f_L=I$  on $f_RA$ because $f_R f_L(f_R(w)) = f_R(f_L f_R(w)) = f_R(w)$
for all $w\in A$.
Since $x\in f_L(f_RA \cap g_RB)$, we know
$f_R(x) \in f_Rf_L(f_RA \cap g_RB) = f_RA \cap g_RB$, 
which implies $f_R(x) = g_R(b_x)$ for some $b_x\in B$.
Similarly, we know $f_R(y) = g_R(b_y)$ for some
$b_y \in B$.
So, we have $g_Lg_R(b_x) = g_Lf_R(x)$ and $g_Lg_R(b_y) = g_Lf_R(y)$.
If we assume $g_Lf_R(x) = g_Lf_R(y)$, then we have $g_Lg_R(b_x) = g_Lg_R(b_y)$.
Since $g_Lg_R$ is injective on $B$, we know $b_x = b_y$.
Thus
$f_R(x) = g_R(b_x) = g_R(b_y) = f_R(y)$, which implies $f_Lf_R(x) = f_Lf_R(y)$.
Since $f_Lf_R = I$ on $A$, we know $x = y$.
Thus $g_Lf_R$ is
injective on $f_L(f_RA \cap g_RB)$.
\end{proof}

\newpage
\section{A Linear Rank Inequality for fields of characteristic other than 3}

In this section, we use the known T8 matroid to construct a ``T8 network'',
and then in turn we use the T8 network to guide a construction of a 
``T8 linear rank inequality'' that is shown to hold for all vector spaces
having finite scalar fields of characteristic not equal to $3$.
Then we show that the T8 inequality does not necessarily hold when such scalar
fields have characteristic $3$.
Finally, we determine the exact coding capacity of the
T8 network and its linear coding capacity over finite field alphabets
of characteristic $3$,
as well as a linear capacity upper bound for finite field alphabets
whose characteristic is not $3$.

The T8 matroid \cite{matroidBook} is a vector matroid which
is represented by the following matrix, 
where column dependencies are over characteristic 3:

$$\bordermatrix{& A & B & C & D & W & X & Y & Z \cr
                    & 1 & 0 & 0 & 0 & 0 & 1 & 1 & 1\cr
                    & 0 & 1 & 0 & 0 & 1 & 0 & 1 & 1\cr
                    & 0 & 0 & 1 & 0 & 1 & 1 & 0 & 1\cr
                    & 0 & 0 & 0 & 1 & 1 & 1 & 1 & 0 }. $$

The T8 matroid is representable over a field if and only if the field is
of characteristic 3.
Figure~\ref{fig:T8} is a network whose dependencies and
independencies are consistent with the T8 matroid.
It was designed by the
construction process described in \cite{nonshannon}, and we will refer to it as
the T8 network.
Theorem~\ref{thm:T8} uses the T8 network as a guide to
derive a linear rank inequality valid for every characteristic except for 3.
We refer to the inequality in the following theorem as the \textit{T8 linear rank inequality}.

\begin{figure}
\begin{center}
\psfrag{n1}{\LARGE $n_{1}$}
\psfrag{n2}{\LARGE $n_{2}$}
\psfrag{n3}{\LARGE $n_{3}$}
\psfrag{n4}{\LARGE $n_{4}$}
\psfrag{n5}{\LARGE $n_{5}$}
\psfrag{n6}{\LARGE $n_{6}$}
\psfrag{n7}{\LARGE $n_{7}$}
\psfrag{n8}{\LARGE $n_{8}$}
\psfrag{n9}{\LARGE $n_{9}$}
\psfrag{n10}{\LARGE $n_{10}$}
\psfrag{n11}{\LARGE $n_{11}$}
\psfrag{n12}{\LARGE $n_{12}$}
\psfrag{n13}{\LARGE $n_{13}$}
\psfrag{n14}{\LARGE $n_{14}$}
\psfrag{n15}{\LARGE $n_{15}$}
\includegraphics[width=13.5cm]{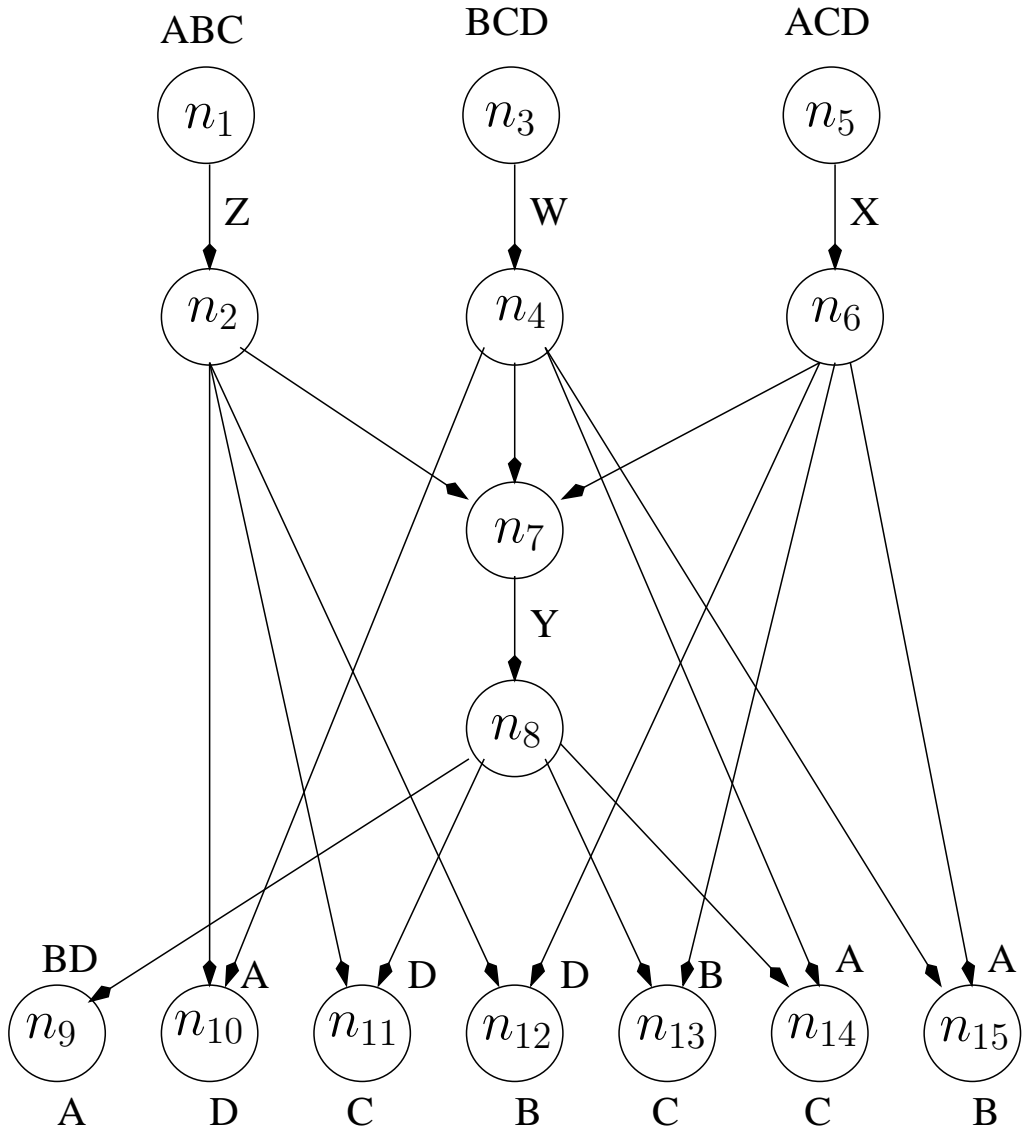}
\end{center}
\caption{ The T8 network has source messages $A,B,C,$ and $D$ generated at hidden source nodes with 
certain hidden out-edges pointing to corresponding displayed nodes
$n_1$, $n_3$, $n_5$, and $n_9$--$n_{15}$
(which are labeled by incoming messages above such nodes).
The nodes $n_9$--$n_{15}$ each demand one message, as labeled below such nodes.
}
\label{fig:T8}
\end{figure}

\begin{thm}\label{thm:T8}
Let $A,B,C,D,W,X,Y$, and $Z$ be subspaces of a vector space $V$ whose scalar field is finite
and of characteristic other than $3$.
Then the following is a linear rank inequality over $V$:
\begin{align*}
H(A) &\le 8H(Z) + 29H(Y) + 3H(X) + 8H(W) -6H(D) -17H(C) - 8H(B) - 17H(A)\\
    &\ \ \ + 55H(Z|A,B,C) + 35H(Y|W,X,Z) + 50H(X|A,C,D) + 49H(W|B,C,D)\\
    &\ \ \ + 18H(A|B,D,Y) + 7H(B|D,X,Z) + H(B|A,W,X) + 7H(C|D,Y,Z) \\
    &\ \ \ + 7H(C|B,X,Y) + 3H(C|A,W,Y) + 6H(D|A,W,Z)\\
    &\ \ \ + 49(H(A) + H(B) + H(C) + H(D) - H(A,B,C,D)).
\end{align*}
\end{thm}

\begin{proof}
By Lemma~\ref{lemma4} we get linear functions:
\begin{center}
\begin{tabular}{ccc}
     $f_1:Z \rightarrow A$, & $f_2:Z \rightarrow B$, & $f_3:Z \rightarrow C$,\\
     $f_4:W \rightarrow B$, & $f_5:W \rightarrow C$, & $f_6:W \rightarrow D$,\\
     $f_7:X \rightarrow A$, & $f_8:X \rightarrow C$, & $f_9:X \rightarrow D$,\\
     $f_{10}:Y \rightarrow Z$, & $f_{11}:Y \rightarrow W$, & $f_{12}:Y \rightarrow X$,\\
     $f_{13}:A \rightarrow B$, & $f_{14}:A \rightarrow D$, & $f_{15}:A \rightarrow Y$,\\
     $f_{16}:D \rightarrow Z$, & $f_{17}:D \rightarrow W$, & $f_{18}:D \rightarrow A$,\\
     $f_{19}:C \rightarrow Z$, & $f_{20}:C \rightarrow Y$, & $f_{21}:C \rightarrow D$,\\
     $f_{22}:B \rightarrow Z$, & $f_{23}:B \rightarrow X$, & $f_{24}:B \rightarrow D$,\\
     $f_{25}:C \rightarrow Y$, & $f_{26}:C \rightarrow X$, & $f_{27}:C \rightarrow B$,\\
     $f_{28}:C \rightarrow Y$, & $f_{29}:C \rightarrow W$, & $f_{30}:C \rightarrow A$,\\
     $f_{31}:B \rightarrow W$, & $f_{32}:B \rightarrow X$, & $f_{33}:B \rightarrow A$
\end{tabular}
\end{center}
such that
\begin{align}
    f_1 + f_2 + f_3 &= I \mbox{ on a subspace of $Z$ of codimension $H(Z|A,B,C)$}\label{eq:1}\\
    f_4 + f_5 + f_6 &= I \mbox{ on a subspace of $W$ of codimension $H(W|B,C,D)$}\label{eq:2}\\
    f_7 + f_8 + f_9 &= I \mbox{ on a subspace of $X$ of codimension $H(X|A,C,D)$}\label{eq:3}\\
    f_{10} + f_{11} + f_{12} &= I \mbox{ on a subspace of $Y$ of codimension $H(Y|W,X,Z)$}\label{eq:4}\\
    f_{13} + f_{14} + f_{15} &= I \mbox{ on a subspace of $A$ of codimension $H(A|B,D,Y)$}\label{eq:5}\\
    f_{16} + f_{17} + f_{18} &= I \mbox{ on a subspace of $D$ of codimension $H(D|A,W,Z)$}\label{eq:6}\\
    f_{19} + f_{20} + f_{21} &= I \mbox{ on a subspace of $C$ of codimension $H(C|D,Y,Z)$}\label{eq:7}\\
    f_{22} + f_{23} + f_{24} &= I \mbox{ on a subspace of $B$ of codimension $H(B|D,X,Z)$}\label{eq:8}\\
    f_{25} + f_{26} + f_{27} &= I \mbox{ on a subspace of $C$ of codimension $H(C|B,X,Y)$}\label{eq:9}\\
    f_{28} + f_{29} + f_{30} &= I \mbox{ on a subspace of $C$ of codimension $H(C|A,W,Y)$}\label{eq:10}\\
    f_{31} + f_{32} + f_{33} &= I \mbox{ on a subspace of $B$ of codimension $H(B|A,W,X)$}.\label{eq:11}
\end{align}
Now let
\begin{align*}
f_A &\triangleq f_7f_{12} + f_1f_{10}\\
f_B &\triangleq f_4f_{11} + f_2f_{10}\\
f_C &\triangleq f_8f_{12} + f_5f_{11} + f_3f_{10}\\
f_D &\triangleq f_9f_{12} + f_6f_{11}.
\end{align*}
Combining the functions we obtained from Lemma~\ref{lemma4} gives new functions:
\begin{align*}
    f_A\Compose f_{15}                               &: A \rightarrow A \\
    f_B\Compose f_{15} + f_{13}                      &: A \rightarrow B\\
    f_C\Compose f_{15}                               &: A \rightarrow C\\
    f_D\Compose f_{15} + f_{14}                      &: A \rightarrow D.
\end{align*}
Using \eqref{eq:1} - \eqref{eq:5}, Lemma~\ref{lemma1}, and Lemma~\ref{lemma3}
we know the sum of these functions is equal to $I$ on a subspace of $A$ of
codimension at most $H(Z|A,B,C) + H(W|B,C,D) + H(X|A,C,D) + H(Y|W,X,Z) + H(A|B,D,Y)$.

Applying Lemma~\ref{lemma6} and Lemma~\ref{lemma1} to 
$f_A\Compose f_{15} - I$, $f_B\Compose f_{15} + f_{13}$, $f_C\Compose f_{15}$, and $f_D\Compose f_{15} + f_{14}$ 
we get a subspace $\overline{A}$ of $A$ of codimension at most
\begin{align*}
    \Delta_{\overline{A}} &= H(Z|A,B,C) + H(W|B,C,D) + H(X|A,C,D) + H(Y|W,X,Z) + H(A|B,D,Y)\\
    & + H(A) + H(B) + H(C) + H(D) - H(A,B,C,D)
\end{align*}
on which

\begin{align}
    f_A\Compose f_{15}                               &= I \label{eq:Abar1}\\
    f_B\Compose f_{15} + f_{13}                      &= 0 \label{eq:Abar2}\\
    f_C\Compose f_{15}                               &= 0 \label{eq:Abar3}\\
    f_D\Compose f_{15} + f_{14}                      &= 0. \label{eq:Abar4}
\end{align}
%
To see how the T8 network is used as a guide, 
consider receiver node $n_9$, which demands $A$.
Let 
$M_1$, 
$M_7$, 
$M_{10}$, 
$M_{12}$, 
$M_{15}$ 
be matrices corresponding to the transformations along the edges
$(A,Z)$,
$(A,X)$,
$(Z,Y)$,
$(X,Y)$,
$(Y,A)$,
respectively.
Using algebra to solve for $A$ one deduces that
\begin{align*}
M_{15}M_{10}M_{1} + M_{15}M_{12}M_{7} &= I.
\end{align*}
Equation \eqref{eq:Abar1} was designed to model this property.

Similarly, we get a subspace $\overline{B}$ of $B$ of codimension at most
\begin{align*}
    \Delta_{\overline{B}} &=  H(Z|A,B,C) + H(X|A,C,D) + H(B|D,X,Z)\\
     &\ \ \ + H(A) + H(B) + H(C) + H(D) - H(A,B,C,D)
\end{align*}
on which

\begin{align}
    f_7\Compose f_{23}  + f_1\Compose f_{22}                      &=  0 \label{eq:Bbar1}\\
    f_2\Compose f_{22}                                         &=  I \label{eq:Bbar2}\\
    f_8\Compose f_{23} + f_3\Compose f_{22}                       &=  0 \label{eq:Bbar3}\\
    f_{24} + f_9\Compose f_{23}                                &=  0. \label{eq:Bbar4}
\end{align}
We get a subspace $\widehat{B}$ of $B$ of codimension at most
\begin{align*}
    \Delta_{\widehat{B}} &=  H(W|B,C,D) + H(X|A,C,D) + H(B|A,W,X)\\
     &\ \ \ + H(A) + H(B) + H(C) + H(D) - H(A,B,C,D)
\end{align*}
on which
\begin{align}
    f_{33}  + f_{7}\Compose f_{32}                             &=  0 \label{eq:Bhat1}\\
    f_4\Compose f_{31}                                         &=  I \label{eq:Bhat2}\\
    f_8\Compose f_{32} + f_5\Compose f_{31}                       &=  0 \label{eq:Bhat3}\\
    f_9\Compose f_{32} + f_6\Compose f_{31}                       &=  0. \label{eq:Bhat4}
\end{align}
We get a subspace $\overline{C}$ of $C$ of codimension at most
\begin{align*}
    \Delta_{\overline{C}} &=  2H(Z|A,B,C) + H(W|B,C,D) + H(X|A,C,D) + H(Y|W,X,Z) + H(C|D,Y,Z)\\
     &\ \ \ + H(A) + H(B) + H(C) + H(D) - H(A,B,C,D)
\end{align*}
on which

\begin{align}
    f_A\Compose f_{20} + f_1\Compose f_{19}                        &=  0 \label{eq:Cbar1}\\
    f_B\Compose f_{20} + f_2\Compose f_{19}                        &=  0 \label{eq:Cbar2}\\
    f_C\Compose f_{20}  + f_3\Compose f_{19}     &=  I \label{eq:Cbar3}\\
    f_D\Compose f_{20} + f_{21}                                 &=  0. \label{eq:Cbar4}
\end{align}
We get a subspace $\widehat{C}$ of $C$ of codimension at most
\begin{align*}
    \Delta_{\widehat{C}} &=  H(Z|A,B,C) + H(W|B,C,D) + 2H(X|A,C,D) + H(Y|W,X,Z) + H(C|B,X,Y)\\
     &\ \ \ + H(A) + H(B) + H(C) + H(D) - H(A,B,C,D)
\end{align*}
on which

\begin{align}
    f_A\Compose f_{25} + f_7\Compose f_{26}                        &=  0 \label{eq:Chat1}\\
    f_B\Compose f_{25} + f_{27}                                 &=  0 \label{eq:Chat2}\\
    f_C\Compose f_{25}  + f_8\Compose f_{26}     &=  I \label{eq:Chat3}\\
    f_D\Compose f_{25} + f_9\Compose f_{26}                        &=  0. \label{eq:Chat4}
\end{align}
We get a subspace $\widetilde{C}$ of $C$ of codimension at most
\begin{align*}
    \Delta_{\widetilde{C}} &=  H(Z|A,B,C) + 2H(W|B,C,D) + H(X|A,C,D) + H(Y|W,X,Z) + H(C|A,W,Y)\\
     &\ \ \ + H(A) + H(B) + H(C) + H(D) - H(A,B,C,D)
\end{align*}
on which

\begin{align}
    f_A\Compose f_{28} + f_{30}                                 &=  0 \label{eq:Ctilde1}\\
    f_B\Compose f_{28} + f_4\Compose f_{29}                        &=  0 \label{eq:Ctilde2}\\
    f_C\Compose f_{28}  + f_5\Compose f_{29}     &=  I \label{eq:Ctilde3}\\
    f_D\Compose f_{28} + f_6\Compose f_{29}                        &=  0. \label{eq:Ctilde4}
\end{align}
We get a subspace $\overline{D}$ of $D$ of codimension at most
\begin{align*}
    \Delta_{\overline{D}} &=  H(Z|A,B,C) + H(W|B,C,D) + H(D|A,W,Z)\\
     &\ \ \ + H(A) + H(B) + H(C) + H(D) - H(A,B,C,D)
\end{align*}
on which

\begin{align}
    f_{18}  + f_{1}\Compose f_{16}                             &=  0 \label{eq:Dbar1}\\
    f_4\Compose f_{17} + f_2\Compose f_{16}                       &=  0 \label{eq:Dbar2}\\
    f_5\Compose f_{17} + f_3\Compose f_{16}                       &=  0 \label{eq:Dbar3}\\
    f_6\Compose f_{17}                                         &=  I. \label{eq:Dbar4}
\end{align}

First notice that \eqref{eq:Abar1} implies
\begin{align}
&f_{15} \mbox{ is injective on }\overline{A}. \label{p2:f15}
\end{align}
We need to define a subspace of $\overline{A}$ on which $f_{13}$ and $f_{14}$
are injective.
The justifications can be found on \eqref{p2:f13} and \eqref{p2:f14}.
Let
\begin{align*}
\overline{C}^\ast &\triangleq f_3(f_{19}(\overline{C} \cap f_{20}^{-1}f_{15}\overline{A}) \cap f_{22}\overline{B}) \subseteq \overline{C}\\
\widetilde{C}^\ast    &\triangleq f_5(f_{29}(\widetilde{C} \cap f_{28}^{-1}f_{15}\overline{A}) \cap f_{17}\overline{D}) \subseteq \widetilde{C}\\
\overline{A}^\ast &\triangleq f_A(f_{15}\overline{A} \cap f_{20}\overline{C}^\ast \cap f_{28}\widetilde{C}^\ast) \subseteq \overline{A}.
\end{align*}
To justify why $\overline{C}^\ast \subseteq \overline{C}$, by \eqref{eq:Abar3}
we know $f_Cf_{15} = 0$ on $\overline{A}$
and by \eqref{eq:Cbar3} we know $f_Cf_{20} + f_3f_{19} = I$.
Thus for each $\overline{c} \in \overline{C} \cap f_{20}^{-1}f_{15}\overline{A}$, 
we have $f_Cf_{20} = 0$ on $\overline{C}$
which gives
\begin{align}
f_3f_{19} &=  I \ \mbox{ on\ $\overline{C} \cap f_{20}^{-1}f_{15}\overline{A}$}. \label{p2:f3f19}
\end{align}
Using (\ref{eq:Abar3}) and (\ref{eq:Ctilde3}) we have
\begin{align}
f_5f_{29} &=  I \ \mbox{ on\ $\widetilde{C} \cap f_{28}^{-1}f_{15}\overline{A}$}. \label{p2:f5f29}
\end{align}
Using (\ref{eq:Abar3}) and (\ref{eq:Chat3}) we have
\begin{align}
f_8f_{26} &=  I \ \mbox{ on\ $\widehat{C} \cap f_{25}^{-1}f_{15}\overline{A}$}. \label{p2:f8f26}
\end{align}

We are now going to show $f_{13}$ is injective on $\overline{A}^\ast$.
First we
need to apply Lemma~\ref{lem Inj} to show $f_2f_{19}$ is injective on
$\overline{C}^\ast$ and then again to show $f_Bf_{15}$ is injective on
$\overline{A}^\ast$.
By (\ref{eq:Bbar2}) and (\ref{p2:f3f19}), we know
$f_2f_{22}$ is injective on $\overline{B}$ and $f_3f_{19} = I$ on $\overline{C} \cap f_{20}^{-1}f_{15}\overline{A}$.
So, we can apply Lemma~\ref{lem Inj} by
letting $g_L = f_2$, $g_R = f_{22}$, $f_L = f_3$, and $f_R = f_{19}$ to get
that $f_2f_{19}$ is injective on $\overline{C}^\ast$.
Then using
(\ref{eq:Cbar2}), we know $f_Bf_{20}$ is injective on $\overline{C}^\ast$.
Now
we can apply Lemma~\ref{lem Inj} again by using the fact that $f_Af_{15} = I$
on $\overline{A}$ and by letting $g_L = f_B$, $g_R = f_{20}$, $f_L = f_A$, and
$f_R = f_{15}$ to get $f_Bf_{15}$ is injective on $\overline{A}^\ast$.
Thus by
(\ref{eq:Abar2}),
\begin{align}
& f_{13} \mbox{ is injective on $\overline{A}^\ast$.}\label{p2:f13}
\end{align}

Similarly, we are going to show $f_{14}$ is injective on
$\overline{A}^\ast$.
We will first apply Lemma~\ref{lem Inj} to show
$f_6f_{29}$ is injective on $\widetilde{C}^\ast$ and then again to show
$f_Df_{15}$ is injective on $\overline{A}^\ast$.
By (\ref{eq:Dbar4}) and
(\ref{p2:f5f29}), we know $f_6f_{17}$ is injective on $\overline{D}$ and
$f_5f_{29} = I$ on $\widetilde{C} \cap f_{28}^{-1}f_{15}\overline{A}$.
So, we
can apply Lemma~\ref{lem Inj} by letting $g_L = f_6$, $g_R = f_{17}$, 
$f_L = f_5$, and $f_R = f_{29}$ to get that $f_6f_{29}$ is injective on
$\widetilde{C}^\ast$.
Then using (\ref{eq:Ctilde4}), we know $f_Df_{28}$ is
injective on $\widetilde{C}^\ast$.
Now we can apply Lemma~\ref{lem Inj} again
by using the fact that $f_Af_{15} = I$ on $\overline{A}$ and by letting 
$g_L = f_D$, $g_R = f_{28}$, $f_L = f_A$, and $f_R = f_{15}$ to get $f_Df_{15}$ is
injective on $\overline{A}^\ast$.
Thus by (\ref{eq:Abar4}),
\begin{align}
& f_{14}\mbox{ is injective on $\overline{A}^\ast$.}\label{p2:f14}
\end{align}

Now we are going to find an upper bound for $\Codim_A(\overline{A}^\ast)$.
First
we need to find upper bounds for $\Codim_C(\overline{C}^\ast)$ and
$\Codim_C(\widetilde{C}^\ast)$.
Using (\ref{p2:f15}) to show
$\Dim(f_{15}\overline{A}) = \Dim(\overline{A})$, and again using Lemma
\ref{lemma1} and Lemma~\ref{lemma3}, we have
\begin{align}
\Codim_C(\overline{C}^\ast) &=  H(C) - \Dim(\overline{C}^\ast)\nonumber\\
&=  H(C) - \Dim(f_3(f_{19}(\overline{C} \cap f_{20}^{-1}f_{15}\overline{A}) \cap f_{22}\overline{B})) \nonumber\\
&=  H(C) - \Dim(f_{19}(\overline{C} \cap f_{20}^{-1}f_{15}\overline{A}) \cap f_{22}\overline{B})\nonumber\\
&=  H(C) - H(Z) + \Codim_Z(f_{19}(\overline{C} \cap f_{20}^{-1}f_{15}\overline{A}) \cap f_{22}\overline{B})\nonumber\\
&\le  H(C) - H(Z) + \Codim_Z(f_{19}(\overline{C} \cap f_{20}^{-1}f_{15}\overline{A})) + \Codim_Z(f_{22}\overline{B})\nonumber\\
&=  H(C) - H(Z) + H(Z) - \Dim(f_{19}(\overline{C} \cap f_{20}^{-1}f_{15}\overline{A})) + H(Z) - \Dim(f_{22}\overline{B})\nonumber\\
&=  H(C) + H(Z) - \Dim(\overline{C} \cap f_{20}^{-1}f_{15}\overline{A}) - \Dim(\overline{B})\nonumber\\
&=  H(C) + H(Z) - H(C) + \Codim_C(\overline{C} \cap f_{20}^{-1}f_{15}\overline{A}) - H(B) + \Codim_B(\overline{B})\nonumber\\
&=  H(Z) - H(B) + \Codim_C(\overline{C} \cap f_{20}^{-1}f_{15}\overline{A}) + \Codim_B(\overline{B})\nonumber\\
&\le  H(Z) - H(B) + \Delta_{\overline{C}} + \Codim_C( f_{20}^{-1}f_{15}\overline{A}) + \Delta_{\overline{B}}\nonumber\\
&\le  H(Z) - H(B) + \Delta_{\overline{C}} + \Codim_Y( f_{15}\overline{A}) + \Delta_{\overline{B}}\nonumber\\
&\le  H(Z) - H(B) + \Delta_{\overline{C}} + H(Y) - \Dim( f_{15}\overline{A}) + \Delta_{\overline{B}}\nonumber\\
&=  H(Z) - H(B) + \Delta_{\overline{C}} + H(Y) - \Dim(\overline{A}) + \Delta_{\overline{B}}\nonumber\\
&=  H(Z) - H(B) + \Delta_{\overline{C}} + H(Y) - H(A) + \Codim_A(\overline{A}) + \Delta_{\overline{B}}\nonumber\\
&\le  H(Z) - H(B) + H(Y) - H(A) +\Delta_{\overline{C}} + \Delta_{\overline{A}} + \Delta_{\overline{B}}
\end{align}
\begin{align}
\Codim_C(\widetilde{C}^\ast) &=  H(C) - \Dim(\widetilde{C}^\ast)\nonumber\\
&=  H(C) - \Dim(f_5(f_{29}(\widetilde{C} \cap f_{28}^{-1}f_{15}\overline{A}) \cap f_{17}\overline{D}))\nonumber\\
&=  H(C) - \Dim(f_{29}(\widetilde{C} \cap f_{28}^{-1}f_{15}\overline{A}) \cap f_{17}\overline{D})\nonumber\\
&=  H(C) - H(W) + \Codim_W(f_{29}(\widetilde{C} \cap f_{28}^{-1}f_{15}\overline{A}) \cap f_{17}\overline{D})\nonumber\\
&\le  H(C) - H(W) + \Codim_W(f_{29}(\widetilde{C} \cap f_{28}^{-1}f_{15}\overline{A})) + \Codim_W(f_{17}\overline{D})\nonumber\\
&=  H(C) - H(W) + H(W) - \Dim(f_{29}(\widetilde{C} \cap f_{28}^{-1}f_{15}\overline{A})) + H(W) - \Dim(f_{17}\overline{D})\nonumber\\
&=  H(C) + H(W) - \Dim(\widetilde{C} \cap f_{28}^{-1}f_{15}\overline{A}) - \Dim(\overline{D})\nonumber\\
&=  H(C) + H(W) - H(C) + \Codim_C(\widetilde{C} \cap f_{28}^{-1}f_{15}\overline{A}) - H(D) + \Codim_D(\overline{D})\nonumber\\
&=  H(W) - H(D) + \Codim_C(\widetilde{C} \cap f_{28}^{-1}f_{15}\overline{A}) + \Codim_D(\overline{D})\nonumber\\
&\le  H(W) - H(D) + \Delta_{\widetilde{C}} + \Codim_C( f_{28}^{-1}f_{15}\overline{A}) + \Delta_{\overline{D}}\nonumber\\
&\le  H(W) - H(D) + \Delta_{\widetilde{C}} + \Codim_Y( f_{15}\overline{A}) + \Delta_{\overline{D}}\nonumber\\
&=  H(W) - H(D) + \Delta_{\widetilde{C}} + H(Y) - \Dim( f_{15}\overline{A}) + \Delta_{\overline{D}}\nonumber\\
&=  H(W) - H(D) + \Delta_{\widetilde{C}} + H(Y) - \Dim( \overline{A}) + \Delta_{\overline{D}}\nonumber\\
&=  H(W) - H(D) + \Delta_{\widetilde{C}} + H(Y) - H(A) +  \Codim_A(\overline{A}) + \Delta_{\overline{D}}\nonumber\\
&\le  H(W) - H(D) + H(Y) - H(A) + \Delta_{\widetilde{C}} +  \Delta_{\overline{A}} + \Delta_{\overline{D}}.
\end{align}
In the justification for (\ref{p2:f13}), we concluded that $f_Bf_{20}$ is
injective on $\overline{C}^\ast$, which implies $f_{20}$ is injective on
$\overline{C}^\ast$.
In the justification for (\ref{p2:f14}), we concluded that
$f_Df_{28}$ is injective on $\widetilde{C}^\ast$, which implies $f_{28}$ is
injective on $\widetilde{C}^\ast$.
These facts combined with (\ref{p2:f15})
will be used to arrive on line (\ref{app:inj1}).
\begin{align}
\Codim_A(\overline{A}^\ast) &=  H(A) - \Dim(f_A(f_{15}\overline{A} \cap f_{20}\overline{C}^\ast \cap f_{28}\widetilde{C}^\ast))\nonumber\\
&=  H(A) - \Dim(f_{15}\overline{A} \cap f_{20}\overline{C}^\ast \cap f_{28}\widetilde{C}^\ast)\nonumber\\
&=  H(A) - H(Y) + \Codim_Y(f_{15}\overline{A} \cap f_{20}\overline{C}^\ast \cap f_{28}\widetilde{C}^\ast)\nonumber\\
&\le  H(A) - H(Y) + \Codim_Y(f_{15}\overline{A}) + \Codim_Y(f_{20}\overline{C}^\ast) + \Codim_Y(f_{28}\widetilde{C}^\ast)\nonumber\\
&=  H(A) - H(Y) + H(Y) - \Dim(f_{15}\overline{A}) + H(Y) - \Dim(f_{20}\overline{C}^\ast) \nonumber\\
&\ \ \  + H(Y) - \Dim(f_{28}\widetilde{C}^\ast)\nonumber\\
&=  H(A) + 2H(Y) - \Dim(\overline{A})  - \Dim(\overline{C}^\ast) - \Dim(\widetilde{C}^\ast)\label{app:inj1}\\
&=  H(A) + 2H(Y) - H(A) + \Codim_A(\overline{A})  - H(C) + \Codim_C(\overline{C}^\ast) \nonumber\\
&\ \ \  - H(C) + \Codim_C(\widetilde{C}^\ast)\nonumber\\
&=  2H(Y) - 2H(C) + \Codim_A(\overline{A})  + \Codim_C(\overline{C}^\ast) + \Codim_C(\widetilde{C}^\ast)\nonumber\\
&\le  2H(Y) - 2H(C) + \Delta_{\overline{A}} \nonumber\\
&\ \ \ + H(Z) - H(B) + H(Y) - H(A) +\Delta_{\overline{C}} + \Delta_{\overline{A}} + \Delta_{\overline{B}}\nonumber\\
&\ \ \ + H(W) - H(D) + H(Y) - H(A) + \Delta_{\widetilde{C}} +  \Delta_{\overline{A}} + \Delta_{\overline{D}}\nonumber\\
&=  H(W) + 4H(Y) + H(Z) - 2H(A) - H(B) -2H(C) - H(D) \nonumber\\
&\ \ \ + 3\Delta_{\overline{A}} + \Delta_{\overline{B}} + \Delta_{\overline{C}} + \Delta_{\widetilde{C}}+ \Delta_{\overline{D}}\nonumber\\
&\triangleq \Delta_{\overline{A}^\ast}. \label{app:Abarstar}
\end{align}

Let $t \in A$.
We will next make a collection of assumptions on $t$ in (\ref{p2:ass1})--(\ref{p2:ass6}).
Each such assumption gives rise to an upper bound on the codimension of a particular subspace of $A$.
The justification of these upper bounds will be given in what follows.
Ultimately, we will show that these assumptions imply that $3t=0$ and
thus for field characteristics other than $3$,
no nonzero $t$ can satisfy this condition.
This in turn implies that the codimension of the intersection of the subspaces of $A$ in the upper bounds
of (\ref{p2:ass1})--(\ref{p2:ass6})
must be at least as big as the dimension of $A$,
which then yields the desired inequality.

\begin{align}
&\mbox{We will assume } t \in \overline{A}^\ast. 
 \mbox{ This is true on a subspace of $A$ of codimension at most } \Delta_{\overline{A}^\ast}. \label{p2:ass1}\\
&\mbox{We will assume } f_{10}f_{15}t \in f_{19}(\overline{C} \cap f_{20}^{-1}f_{15}\overline{A}^\ast). 
 \mbox{ This is true on a subspace of $A$ of  } \nonumber\\
&\qquad\qquad \mbox{ codimension at most } H(Z) - H(C) + H(Y) - H(A) + \Delta_{\overline{C}} + \Delta_{\overline{A}^\ast}.\label{p2:ass2}\\
&\mbox{We will assume } f_{11}f_{15}t \in f_{29}(\widetilde{C} \cap f_{28}^{-1}f_{15}\overline{A}^\ast). 
 \mbox{ This is true on a subspace of $A$ of } \nonumber\\
&\qquad\qquad \mbox{ codimension at most } H(W) -H(C) + H(Y) - H(A) + \Delta_{\widetilde{C}} + \Delta_{\overline{A}^\ast}.\label{p2:ass3}\\
&\mbox{We will assume } f_{12}f_{15}t \in f_{26}(\widehat{C} \cap f_{25}^{-1}f_{15}\overline{A}^\ast). 
 \mbox{ This is true on a subspace of $A$ of } \nonumber\\
&\qquad\qquad \mbox{ codimension at most } H(X)- H(C) + H(Y) - H(A) + \Delta_{\widehat{C}} + \Delta_{\overline{A}^\ast}.\label{p2:ass4}\\
&\mbox{We will assume } f_{10}f_{15}t \in f_{22}(\overline{B} \cap f_{23}^{-1}f_{26}(\widehat{C} \cap f_{25}^{-1}f_{15}\overline{A}^\ast)). 
 \mbox{ This is true on a subspace } \nonumber\\
&\qquad\qquad \mbox{ of $A$ of codimension at most }\nonumber\\
&\qquad\qquad H(Z) - H(B) + H(X) - H(C) + H(Y) - H(A) + \Delta_{\overline{A}^\ast} + \Delta_{\overline{B}} + \Delta_{\widehat{C}}.\label{p2:ass5}\\
&\mbox{We will assume } f_{11}f_{15}t \in f_{31}(\widehat{B} \cap f_{32}^{-1}f_{26}(\widehat{C} \cap f_{25}^{-1}f_{15}\overline{A}^\ast)).
 \mbox{ This is true on a subspace } \nonumber\\
&\qquad\qquad\mbox{ of $A$ of codimension at most }\nonumber\\
& \qquad \qquad H(W) - H(B) + H(X) - H(C) + H(Y) - H(A) + \Delta_{\overline{A}^\ast} + \Delta_{\widehat{B}} + \Delta_{\widehat{C}}. \label{p2:ass6}
\end{align}

To justify (\ref{p2:ass2}), first we know $f_{19}$ is injective on
$\overline{C} \cap f_{20}^{-1}f_{15}\overline{A}^\ast$ by
(\ref{p2:f3f19}).
Then by Lemma~\ref{lemma3}, we know 
$f_{10}f_{15}t \in f_{19}(\overline{C} \cap f_{20}^{-1}f_{15}\overline{A}^\ast)$ 
on a subspace of $A$ of codimension at most 
$H(Z) - H(C) + \Codim_C(\overline{C} \cap f_{20}^{-1}f_{15}\overline{A}^\ast)$.
By Lemma~\ref{lemma1}, we know
\begin{align*}
\Codim_C(\overline{C} \cap f_{20}^{-1}f_{15}\overline{A}^\ast) &\le  \Delta_{\overline{C}} + \Codim_C(f_{20}^{-1}f_{15}\overline{A}^\ast).
\end{align*}
Then using Lemma~\ref{lemma3} and (\ref{p2:f15}), we know
\begin{align}
\Codim_C(\overline{C} \cap f_{20}^{-1}f_{15}\overline{A}^\ast) &\le   \Delta_{\overline{C}} + \Codim_Y(f_{15}\overline{A}^\ast) \nonumber\\
                                                       &=      \Delta_{\overline{C}} + H(Y) - \Dim(f_{15}\overline{A}^\ast) \nonumber\\
                                                       &=      \Delta_{\overline{C}} + H(Y) - \Dim(\overline{A}^\ast) \nonumber\\
                                                       &\le      \Delta_{\overline{C}} + H(Y) - H(A) + \Delta_{\overline{A}^\ast}.
\end{align}
So, we have $f_{10}f_{15}t \in f_{19}(\overline{C} \cap f_{20}^{-1}f_{15}\overline{A}^\ast)$ on a subspace of $A$ of codimension at most 
$H(Z) - H(C) + H(Y) - H(A) + \Delta_{\overline{C}} + \Delta_{\overline{A}^\ast}$.

To justify (\ref{p2:ass3}), first we know $f_{29}$ is injective on
$\widetilde{C} \cap f_{28}^{-1}f_{15}\overline{A}^\ast$ by
(\ref{p2:f5f29}).
Then by Lemma~\ref{lemma3}, 
we know $f_{11}f_{15}t \in f_{29}(\widetilde{C} \cap f_{28}^{-1}f_{15}\overline{A}^\ast)$ on a subspace of
$A$ of codimension at most $H(Z) - H(C) + \Codim_C(\widetilde{C} \cap f_{28}^{-1}f_{15}\overline{A}^\ast)$.
By Lemma~\ref{lemma1}, we know
\begin{align*}
\Codim_C(\widetilde{C} \cap f_{28}^{-1}f_{15}\overline{A}^\ast) &\le  \Delta_{\widetilde{C}} + \Codim_C(f_{28}^{-1}f_{15}\overline{A}^\ast).
\end{align*}
Then using Lemma~\ref{lemma3} and (\ref{p2:f15}), we know
\begin{align}
\Codim_C(\widetilde{C} \cap f_{28}^{-1}f_{15}\overline{A}^\ast) &\le   \Delta_{\widetilde{C}} + \Codim_Y(f_{15}\overline{A}^\ast) \nonumber\\
                                                       &=      \Delta_{\widetilde{C}} + H(Y) - \Dim(f_{15}\overline{A}^\ast) \nonumber\\
                                                       &=      \Delta_{\widetilde{C}} + H(Y) - \Dim(\overline{A}^\ast) \nonumber\\
                                                       &\le      \Delta_{\widetilde{C}} + H(Y) - H(A) + \Delta_{\overline{A}^\ast}.
\end{align}
So, we have $f_{11}f_{15}t \in f_{29}(\widetilde{C} \cap f_{28}^{-1}f_{15}\overline{A}^\ast)$ on a subspace of $A$ of codimension at
most $H(Z) - H(C) + H(Y) - H(A) + \Delta_{\widetilde{C}} + \Delta_{\overline{A}^\ast}$.

To justify (\ref{p2:ass4}), first we know $f_{26}$ is injective on 
$\widehat{C} \cap f_{25}^{-1}f_{15}\overline{A}^\ast$ by (\ref{p2:f8f26}).
Then by Lemma
\ref{lemma3}, we know $f_{12}f_{15}t \in f_{26}(\widehat{C} \cap f_{25}^{-1}f_{15}\overline{A}^\ast)$ 
on a subspace of $A$ of codimension at
most $H(Z) - H(C) + \Codim_C(\widehat{C} \cap f_{25}^{-1}f_{15}\overline{A}^\ast)$.
By Lemma~\ref{lemma1}, we know
\begin{align*}
\Codim_C(\widehat{C} \cap f_{25}^{-1}f_{15}\overline{A}^\ast) &\le  \Delta_{\widehat{C}} + \Codim_C(f_{25}^{-1}f_{15}\overline{A}^\ast)
\end{align*}
Then using Lemma~\ref{lemma3} and (\ref{p2:f15}), we know
\begin{align}
\Codim_C(\widehat{C} \cap f_{25}^{-1}f_{15}\overline{A}^\ast) &\le   \Delta_{\widehat{C}} + \Codim_Y(f_{15}\overline{A}^\ast) \nonumber\\
                                                       &=      \Delta_{\widehat{C}} + H(Y) - \Dim(f_{15}\overline{A}^\ast) \nonumber\\
                                                       &=      \Delta_{\widehat{C}} + H(Y) - \Dim(\overline{A}^\ast) \nonumber\\
                                                       &\le      \Delta_{\widehat{C}} + H(Y) - H(A) + \Delta_{\overline{A}^\ast}.
\end{align}
So, we have $f_{12}f_{15}t \in f_{26}(\widehat{C} \cap f_{25}^{-1}f_{15}\overline{A}^\ast)$ on a subspace of $A$ of codimension at
most $H(Z) - H(C) + H(Y) - H(A) + \Delta_{\widehat{C}} + \Delta_{\overline{A}^\ast}$.

To justify (\ref{p2:ass5}), we first know $f_{22}$ is injective on
$\overline{B} \cap f_{23}^{-1}f_{26}(\widehat{C} \cap  f_{25}^{-1}f_{15}\overline{A}^\ast)$ by (\ref{eq:Bbar2}).
Then by Lemma
\ref{lemma3}, we know 
$f_{10}f_{15}t \in f_{22}(\overline{B} \cap f_{23}^{-1}f_{26}(\widehat{C} \cap f_{25}^{-1}f_{15}\overline{A}^\ast))$ 
on a subspace of $A$ of codimension at most 
$H(Z) - H(B) + \Codim_B(\overline{B} \cap  f_{23}^{-1}f_{26}(\widehat{C} \cap f_{25}^{-1}f_{15}\overline{A}^\ast))$.
Now
again we are going to use Lemma~\ref{lemma1}, Lemma~\ref{lemma3}, and
(\ref{p2:f15}).
Also on line (\ref{p2:condf26}) we will use the fact that
$f_{26}$ is injective on $\widehat{C} \cap f_{25}^{-1}f_{15}\overline{A}^\ast$
from (\ref{p2:f8f26}).
\begin{align}
\Codim_B(\overline{B} &\cap f_{23}^{-1}f_{26}(\widehat{C} \cap f_{25}^{-1}f_{15}\overline{A}^\ast)) \le \Delta_{\overline{B}} + \Codim_B(f_{23}^{-1}f_{26}(\widehat{C} \cap f_{25}^{-1}f_{15}\overline{A}^\ast))\nonumber\\
&\le \Delta_{\overline{B}} + \Codim_X(f_{26}(\widehat{C} \cap f_{25}^{-1}f_{15}\overline{A}^\ast))\nonumber\\
&= \Delta_{\overline{B}} + H(X) - \Dim(f_{26}(\widehat{C} \cap f_{25}^{-1}f_{15}\overline{A}^\ast))\nonumber\\
&= \Delta_{\overline{B}} + H(X) - \Dim(\widehat{C} \cap f_{25}^{-1}f_{15}\overline{A}^\ast)\label{p2:condf26}\\
&\le \Delta_{\overline{B}} + H(X) - H(C) + \Codim_C(\widehat{C}) + \Codim_C(f_{25}^{-1}f_{15}\overline{A}^\ast)\nonumber\\
&\le \Delta_{\overline{B}} + H(X) - H(C) + \Delta_{\widehat{C}} + \Codim_Y(f_{15}\overline{A}^\ast)\nonumber\\
&= \Delta_{\overline{B}} + H(X) - H(C) + \Delta_{\widehat{C}} + H(Y) - \Dim(f_{15}\overline{A}^\ast)\nonumber\\
&= \Delta_{\overline{B}} + H(X) - H(C) + H(Y) + \Delta_{\widehat{C}} - \Dim(\overline{A}^\ast)\nonumber\\
&= \Delta_{\overline{B}} + H(X) - H(C) + H(Y) + \Delta_{\widehat{C}} - H(A) + \Codim_A(\overline{A}^\ast)\nonumber\\
&\le \Delta_{\overline{B}} + H(X) - H(C) + H(Y) - H(A) + \Delta_{\widehat{C}} + \Delta_{\overline{A}^\ast}.\label{p2:condf262}
\end{align}
So, we have $f_{10}f_{15}t \in f_{22}(\overline{B} \cap f_{23}^{-1}f_{26}(\widehat{C} \cap f_{25}^{-1}f_{15}\overline{A}^\ast))$ on a
subspace of $A$ of codimension at most 
$H(Z) - H(B) + H(X) - H(C) + H(Y) - H(A) + \Delta_{\overline{A}^\ast} + \Delta_{\overline{B}} + \Delta_{\widehat{C}}$.

To justify (\ref{p2:ass6}), we first know $f_{31}$ is injective on 
$\widehat{B} \cap f_{32}^{-1}f_{26}(\widehat{C} \cap f_{25}^{-1}f_{15}\overline{A}^\ast)$ 
by (\ref{eq:Bhat2}).
Then by Lemma~\ref{lemma3}, we know 
$f_{11}f_{15}t \in f_{31}(\widehat{B} \cap f_{32}^{-1}f_{26}(\widehat{C} \cap  f_{25}^{-1}f_{15}\overline{A}^\ast))$ 
on a subspace of $A$ of codimension at
most $H(W) - H(B) + \Codim_B(\widehat{B} \cap f_{32}^{-1}f_{26}(\widehat{C} \cap  f_{25}^{-1}f_{15}\overline{A}^\ast))$.
Now again we are going to use Lemma
\ref{lemma1} and Lemma~\ref{lemma3},
\begin{align}
\Codim_B(\widehat{B} \cap f_{32}^{-1}f_{26}(\widehat{C} \cap f_{25}^{-1}f_{15}\overline{A}^\ast)) &\le  \Delta_{\widehat{B}} + \Codim_B(f_{32}^{-1}f_{26}(\widehat{C} \cap f_{25}^{-1}f_{15}\overline{A}^\ast))\nonumber\\
&\le  \Delta_{\widehat{B}} + \Codim_X(f_{26}(\widehat{C} \cap f_{25}^{-1}f_{15}\overline{A}^\ast))\nonumber\\
&\le  \Delta_{\widehat{B}} + H(X) - H(C) + H(Y) - H(A) + \Delta_{\widehat{C}} + \Delta_{\overline{A}^\ast}. \nonumber
\end{align}
The last line was derived by copying the argument from (\ref{p2:condf262}).
So, we have 
$f_{11}f_{15}t \in f_{31}(\widehat{B} \cap f_{32}^{-1}f_{26}(\widehat{C} \cap f_{25}^{-1}f_{15}\overline{A}^\ast))$ 
on a subspace of $A$ of codimension at most 
$H(W) - H(B) + H(X) - H(C) + H(Y) - H(A) + \Delta_{\overline{A}^\ast} + \Delta_{\widehat{B}} + \Delta_{\widehat{C}}$.

From (\ref{p2:ass2}) and (\ref{p2:ass5}) we know 
$\exists \overline{c} \in \overline{C}, \overline{b} \in \overline{B}$ such that
\begin{align}
& f_{10}f_{15}t = f_{19}\overline{c} = f_{22}\overline{b} \mbox{ where } f_{20}\overline{c} \in f_{15}\overline{A}^\ast \mbox{ and } f_{23}\overline{b} \in f_{26}(\widehat{C} \cap f_{25}^{-1}f_{15}\overline{A}^\ast). \label{p2:f10}
\end{align}
From (\ref{p2:ass3}) and (\ref{p2:ass6}) we know 
$\exists \widetilde{c} \in \widetilde{C}, \widehat{b} \in \widehat{B}$ such that
\begin{align}
& f_{11}f_{15}t = f_{29}\widetilde{c} = f_{31}\widehat{b} \mbox{ where } f_{28}\widetilde{c} \in f_{15}\overline{A}^\ast \mbox{ and } f_{32}\widehat{b} \in f_{26}(\widehat{C} \cap f_{25}^{-1}f_{15}\overline{A}^\ast). \label{p2:f11}
\end{align}
From (\ref{p2:ass4}) we know 
$\exists \widehat{c} \in \widehat{C}$ such that
\begin{align}
& f_{12}f_{15}t = f_{26}\widehat{c} \mbox{ where } f_{25}\widehat{c} \in f_{15}\overline{A}^\ast. \label{p2:f12}
\end{align}
From (\ref{eq:Abar1}) and (\ref{eq:Abar2}), we know
\begin{align}
f_Bf_{15} &=  -f_{13}    \mbox{ on $\overline{A}$} \nonumber\\
f_B       &=  -f_{13}f_A \mbox{ on $f_{15}\overline{A}$}. \label{p2:fB}
\end{align}
From (\ref{eq:Abar1}) and (\ref{eq:Abar4}), we know
\begin{align}
f_Df_{15} &=  -f_{14}    \mbox{ on $\overline{A}$} \nonumber\\
f_D       &=  -f_{14}f_A \mbox{ on $f_{15}\overline{A}$}. \label{p2:fD}
\end{align}
From (\ref{eq:Abar1}) we have
\begin{align*}
f_7f_{12}f_{15}t + f_1f_{10}f_{15}t &=  t.
\end{align*}
Then (\ref{p2:f12}), (\ref{p2:f10}), (\ref{eq:Chat1}), and (\ref{eq:Cbar1}) give
\begin{align}
f_7f_{12}f_{15}t + f_1f_{10}f_{15}t &=  t \nonumber\\
f_7f_{26}\widehat{c} + f_1f_{19}\overline{c} &=  t\nonumber\\
-f_Af_{25}\widehat{c} - f_Af_{20}\overline{c} &=  t \nonumber\\
f_Af_{25}\widehat{c} + f_Af_{20}\overline{c} &=  -t. \label{p2:I1}
\end{align}
From (\ref{eq:Abar2}) we have
\begin{align*}
f_4f_{11}f_{15}t + f_2f_{10}f_{15}t &=  -f_{13}t.
\end{align*}
Then (\ref{p2:f11}), (\ref{p2:f10}), (\ref{eq:Ctilde2}), and (\ref{eq:Cbar2}) give
\begin{align}
f_4f_{11}f_{15}t + f_2f_{10}f_{15}t &=  -f_{13}t \nonumber\\
f_4f_{29}\widetilde{c} + f_2f_{19}\overline{c} &=  -f_{13}t\nonumber\\
-f_Bf_{28}\widetilde{c}  -f_Bf_{20}\overline{c} &=  -f_{13}t. \nonumber
\end{align}
By (\ref{p2:f11}) and (\ref{p2:f10}), we know 
$f_{28}\widetilde{c} \in f_{15}\overline{A}^\ast$ and $f_{20}\overline{c} \in f_{15}\overline{A}^\ast$.
Now by (\ref{p2:fB}), we have
\begin{align}
-f_Bf_{28}\widetilde{c}  -f_Bf_{20}\overline{c} &=  -f_{13}t \nonumber\\
f_{13}f_Af_{28}\widetilde{c} + f_{13}f_Af_{20}\overline{c} &=  -f_{13}t. \nonumber
\end{align}
Then using (\ref{eq:Abar1}), we know 
$f_Af_{28}\widetilde{c} \in \overline{A}^\ast$ and $f_Af_{20}\overline{c} \in \overline{A}^\ast$.
By
(\ref{p2:f13}), we have
\begin{align}
f_{13}f_Af_{28}\widetilde{c} + f_{13}f_Af_{20}\overline{c} &=  -f_{13}t \nonumber\\
f_Af_{28}\widetilde{c} + f_Af_{20}\overline{c} &=  -t. \label{p2:I2}
\end{align}
From (\ref{eq:Abar4}) we have
\begin{align*}
f_9f_{12}f_{15}t + f_6f_{11}f_{15}t &=  -f_{14}t.
\end{align*}
Then (\ref{p2:f12}), (\ref{p2:f11}), (\ref{eq:Ctilde4}), and (\ref{eq:Chat4}) give
\begin{align}
f_9f_{12}f_{15}t + f_6f_{11}f_{15}t &=  -f_{14}t \nonumber\\
f_9f_{26}\widehat{c} + f_6f_{29}\widetilde{c} &=  -f_{14}t\nonumber\\
-f_Df_{25}\widehat{c} + -f_Df_{28}\widetilde{c} &=  -f_{14}t. \nonumber
\end{align}
By (\ref{p2:f12}) and (\ref{p2:f11}), we know 
$f_{25}\widehat{c} \in f_{15}\overline{A}^\ast$ and $f_{28}\widetilde{c} \in f_{15}\overline{A}^\ast$.
Now by (\ref{p2:fD}), we have
\begin{align}
-f_Df_{25}\widehat{c} + -f_Df_{28}\widetilde{c} &=  -f_{14}t \nonumber\\
f_{14}f_Af_{25}\widehat{c} + f_{14}f_Af_{28}\widetilde{c} &=  -f_{14}t. \nonumber
\end{align}
Then using (\ref{eq:Abar1}), we know 
$f_Af_{25}\widehat{c} \in \overline{A}^\ast$ and $f_Af_{28}\widetilde{c} \in \overline{A}^\ast$.
By (\ref{p2:f14}), we have
\begin{align}
f_{14}f_Af_{25}\widehat{c} + f_{14}f_Af_{28}\widetilde{c} &=  -f_{14}t \nonumber\\
f_Af_{25}\widehat{c} + f_Af_{28}\widetilde{c} &=  -t. \label{p2:I3}
\end{align}
From (\ref{eq:Cbar1}) and (\ref{p2:f3f19}), we know
\begin{align}
f_1f_{19} &=  -f_Af_{20} \mbox{ on $\overline{C}$} \nonumber\\
f_1 &=  -f_Af_{20}f_3\mbox{ on $f_{19}(\overline{C} \cap f_{20}^{-1}f_{15}\overline{A}^\ast)$}. \label{p2:f1}
\end{align}
From (\ref{eq:Chat1}) and (\ref{p2:f8f26}), we know
\begin{align}
f_7f_{26} &=  -f_Af_{25} \mbox{ on $\widehat{C}$} \nonumber\\ 
f_7 &=  -f_Af_{25}f_8\mbox{ on $f_{26}(\widehat{C} \cap f_{25}^{-1}f_{15}\overline{A}^\ast)$}. \label{p2:f7}
\end{align}
From (\ref{eq:Bbar1}), we have
\begin{align}
f_7f_{23}\overline{b} + f_1f_{22}\overline{b} &=  0. \nonumber
\end{align}
By (\ref{p2:f10}), we know $f_{23}\overline{b} \in f_{26}(\widehat{C} \cap f_{25}^{-1}f_{15}\overline{A}^\ast)$.
By (\ref{p2:f10}), we also know
$f_{22}\overline{b} = f_{19}\overline{c}$, which implies 
$f_{22}\overline{b} \in f_{19}(\overline{C} \cap f_{20}^{-1}f_{15}\overline{A}^\ast)$.
Now we can
apply (\ref{p2:f1}) and (\ref{p2:f7}) to give us
\begin{align}
f_7f_{23}\overline{b} + f_1f_{22}\overline{b} &=  0 \nonumber\\
-f_Af_{25}f_8f_{23}\overline{b} - f_Af_{20}f_3f_{22}\overline{b} &=  0. \nonumber
\end{align}
Now using (\ref{eq:Bbar3}), (\ref{p2:f10}), and (\ref{p2:f3f19}), we have
\begin{align}
-f_Af_{25}f_8f_{23}\overline{b} - f_Af_{20}f_3f_{22}\overline{b} &=  0 \nonumber\\
f_Af_{25}f_3f_{22}\overline{b} - f_Af_{20}f_3f_{22}\overline{b} &=  0 \nonumber\\
f_Af_{25}f_3f_{22}\overline{b} &=  f_Af_{20}f_3f_{22}\overline{b} \nonumber\\
f_Af_{25}f_{3}f_{19}\overline{c} &=  f_Af_{20}f_3f_{19}\overline{c}\nonumber\\
f_Af_{25}\overline{c} &=  f_Af_{20}\overline{c}.\label{p2:f20f25}
\end{align}
From (\ref{eq:Chat4}) and (\ref{p2:f8f26}), we know
\begin{align}
f_9f_{26} &=  -f_Df_{25} \mbox{ on $\widehat{C}$} \nonumber\\ 
f_9 &=  -f_Df_{25}f_{8} \mbox{ on $f_{26}(\widehat{C} \cap f_{25}^{-1}f_{15}\overline{A}^\ast)$}.\label{p2:f9}
\end{align}
From (\ref{eq:Ctilde4}) and (\ref{p2:f5f29}), we know
\begin{align}
f_6f_{29} &=  -f_Df_{28} \mbox{ on $\widetilde{C}$} \nonumber\\ 
f_6 &=  -f_Df_{28}f_{5} \mbox{ on $f_{29}(\widetilde{C} \cap f_{28}^{-1}f_{15}\overline{A}^\ast)$}.\label{p2:f6}
\end{align}
From (\ref{eq:Bhat4}), we have
\begin{align}
f_9f_{32}\widehat{b} + f_6f_{31}\widehat{b} &=  0. \nonumber
\end{align}
From (\ref{p2:f11}) we know $f_{31}\widehat{b} = f_{29}\widetilde{c}$ so
$f_{31}\widehat{b} \in f_{29}(\widetilde{C} \cap f_{28}^{-1}f_{15}\overline{A}^\ast)$.
From (\ref{p2:f11}) we also know that
$f_{32}\widehat{b} \in f_{26}(\widehat{C} \cap f_{25}^{-1}f_{15}\overline{A}^\ast)$, 
so (\ref{p2:f9}) and (\ref{p2:f6}) give us
\begin{align}
f_9f_{32}\widehat{b} + f_6f_{31}\widehat{b} &=  0 \nonumber\\
-f_Df_{25}f_8f_{32}\widehat{b} - f_Df_{28}f_5f_{31}\widehat{b} &=  0. \nonumber
\end{align}
From (\ref{p2:f11}), we know $f_{32}\widehat{b} \in f_{26}(\widehat{C} \cap f_{25}^{-1}f_{15}\overline{A}^\ast)$.
From (\ref{p2:f8f26}), we know $f_8f_{26} = I$ on $\widehat{C} \cap f_{25}^{-1}f_{15}\overline{A}^\ast$.
So
$f_8f_{32}\widehat{b} \in f_{25}^{-1}f_{15}\overline{A}^\ast$, which implies
$f_{25}f_8f_{32}\widehat{b} \in f_{15}\overline{A}^\ast$.
By (\ref{p2:f11}) and
(\ref{p2:f5f29}), we know 
$f_{28}f_5f_{31}\widehat{b} = f_{28}f_5f_{29}\widetilde{c} = f_{28}\widetilde{c} \in f_{15}\overline{A}^\ast$.
Now we can apply (\ref{p2:fD}) to give us
\begin{align}
-f_Df_{25}f_8f_{32}\widehat{b} - f_Df_{28}f_5f_{31}\widehat{b} &=  0 \nonumber\\
f_{14}f_Af_{25}f_8f_{32}\widehat{b} + f_{14}f_Af_{28}f_5f_{31}\widehat{b} &=  0. \nonumber
\end{align}
Since we already established that 
$f_{25}f_8f_{32}\widehat{b} \in f_{15}\overline{A}^\ast$ and $f_{28}f_5f_{31}\widehat{b} \in f_{15}\overline{A}^\ast$, 
by (\ref{eq:Abar1}) and (\ref{p2:f14}) we know
\begin{align}
f_{14}f_Af_{25}f_8f_{32}\widehat{b} + f_{14}f_Af_{28}f_5f_{31}\widehat{b} &=  0 \nonumber\\
f_Af_{25}f_8f_{32}\widehat{b} + f_Af_{28}f_5f_{31}\widehat{b} &=  0.\nonumber
\end{align}
Now by (\ref{eq:Bhat3})
\begin{align}
f_Af_{25}f_8f_{32}\widehat{b} + f_Af_{28}f_5f_{31}\widehat{b} &=  0\nonumber\\
-f_Af_{25}f_5f_{31}\widehat{b} + f_Af_{28}f_5f_{31}\widehat{b} &=  0\nonumber\\
f_Af_{25}f_5f_{31}\widehat{b} &=  f_Af_{28}f_5f_{31}\widehat{b}. \nonumber
\end{align}
By (\ref{p2:f11}) and (\ref{p2:f5f29}), we have
\begin{align}
f_Af_{25}f_5f_{31}\widehat{b} &=  f_Af_{28}f_5f_{31}\widehat{b} \nonumber\\
f_Af_{25}f_5f_{29}\widetilde{c} &=  f_Af_{28}f_5f_{29}\widetilde{c} \nonumber\\
f_Af_{25}\widetilde{c} &=  f_Af_{28}\widetilde{c}. \label{p2:f28f25}
\end{align}
Now adding (\ref{p2:I1}), (\ref{p2:I2}), and (\ref{p2:I3}), we have
\begin{align*}
-3t &=  2(f_Af_{20}\overline{c} + f_Af_{25}\widehat{c} + f_Af_{28}\widetilde{c}).
\end{align*}
Now using (\ref{p2:f20f25}) and (\ref{p2:f28f25}) we have
\begin{align}
-3t &=  2(f_Af_{25}\overline{c} + f_Af_{25}\widehat{c} + f_Af_{25}\widetilde{c})\nonumber\\
-3t &=  2f_Af_{25}(\overline{c} + \widehat{c} + \widetilde{c}).\nonumber
\end{align}
By (\ref{p2:f3f19}), (\ref{p2:f5f29}), and (\ref{p2:f8f26}) we know
\begin{align}
-3t &=  2f_Af_{25}(f_3f_{19}\overline{c} + f_8f_{26}\widehat{c} + f_5f_{29}\widetilde{c}).\nonumber
\end{align}
By (\ref{p2:f10}), (\ref{p2:f11}), (\ref{p2:f12}), and (\ref{eq:Abar3}), we have
\begin{align}
-3t &=  2f_Af_{25}(f_3f_{10}f_{15}t + f_8f_{12}f_{15}t + f_5f_{11}f_{15}t)\nonumber\\
-3t &=  2f_Af_{25}(0)\nonumber\\
3t &=  0.
\end{align}
Thus if the field is of characteristic other than 3, then no nonzero $t$ can
satisfy conditions (\ref{p2:ass1})--(\ref{p2:ass6}).
Therefore the sum of the
codimensions given in the assumptions must be at least the dimension of $A$.
So
we have a linear rank inequality for fields of characteristic other than 3:
\begin{align*}
    H(A) &\le   \Delta_{\overline{A}^\ast} + H(Z) - H(C) + H(Y) - H(A) + \Delta_{\overline{C}} + \Delta_{\overline{A}^\ast} \\
    &\ \ \ + H(W) - H(C) + H(Y) - H(A) + \Delta_{\widetilde{C}}+ \Delta_{\overline{A}^\ast} \\
    &\ \ \  + H(X) - H(C) + H(Y) - H(A) + \Delta_{\widehat{C}} + \Delta_{\overline{A}^\ast} \\
    &\ \ \  + H(Z) - H(B) + H(X) - H(C) + H(Y) - H(A) + \Delta_{\overline{A}^\ast} + \Delta_{\overline{B}} + \Delta_{\widehat{C}}\\
    &\ \ \  + H(W) - H(B) + H(X) - H(C) + H(Y) - H(A) + \Delta_{\overline{A}^\ast} + \Delta_{\widehat{B}} + \Delta_{\widehat{C}}\\
    &=  2H(Z) + 5H(Y) + 3H(X) + 2H(W) - 5H(A) -2H(B) - 5H(C)\\
    &\ \ \  + 6\Delta_{\overline{A}^\ast} + \Delta_{\overline{B}} + \Delta_{\widehat{B}} + \Delta_{\overline{C}} + \Delta_{\widetilde{C}} +  3\Delta_{\widehat{C}}\\
    &=  2H(Z) + 5H(Y) + 3H(X) + 2H(W) - 5H(A) -2H(B) - 5H(C)\\
    &\ \ \  + 6(  H(W) + 4H(Y) + H(Z) - 2H(A) - H(B) -2H(C) - H(D) )\\
    &\ \ \  + 6( 3\Delta_{\overline{A}} + \Delta_{\overline{B}} + \Delta_{\overline{C}} + \Delta_{\widetilde{C}}+ \Delta_{\overline{D}} )+ \Delta_{\overline{B}} + \Delta_{\widehat{B}} + \Delta_{\overline{C}} + \Delta_{\widetilde{C}} + 3\Delta_{\widehat{C}}\\
    &=  8H(Z) + 29H(Y) + 3H(X) + 8H(W) -6H(D) -17H(C) - 8H(B) - 17H(A)\\
    &\ \ \  + 18\Delta_{\overline{A}} + 7\Delta_{\overline{B}} + \Delta_{\widehat{B}} + 7\Delta_{\overline{C}} + 7\Delta_{\widetilde{C}} + 3\Delta_{\widehat{C}} + 6\Delta_{\overline{D}}\\
    &=  8H(Z) + 29H(Y) + 3H(X) + 8H(W) -6H(D) -17H(C) - 8H(B) - 17H(A)\\
    &\ \ \  + 55H(Z|A,B,C) + 35H(Y|W,X,Z) + 50H(X|A,C,D) + 49H(W|B,C,D)\\
    &\ \ \  + 18H(A|B,D,Y) + 7H(B|D,X,Z) + H(B|A,W,X) + 7H(C|D,Y,Z) \\
    &\ \ \  + 7H(C|B,X,Y) + 3H(C|A,W,Y) + 6H(D|A,W,Z)\\
    &\ \ \  + 49(H(A) + H(B) + H(C) + H(D) - H(A,B,C,D)).
\end{align*}
\end{proof}

The next theorem demonstates that the inequality in Theorem~\ref{thm:T8}
does not in general hold for vector spaces with finite fields of characteristic $3$.

\begin{thm}\label{thm:T8-char3}
There exists a vector space $V$ with a finite scalar field of characteristic $3$ 
such that the T8 inequality in Theorem~\ref{thm:T8} 
is not a linear rank inequality over $V$.
\end{thm}

\begin{proof}
Let $V$ be the vector space of $4$-dimensional vectors whose components are from the field $GF(3)$,
and define the following subspaces of $V$:
\begin{align*}
    A = \langle (1,0,0,0)\rangle &\ \ \  B = \langle (0,1,0,0)\rangle\\
    C = \langle (0,0,1,0)\rangle &\ \ \  D = \langle (0,0,0,1)\rangle\\
    W = \langle (0,1,1,1)\rangle &\ \ \  X = \langle (1,0,1,1)\rangle\\
    Y = \langle (1,1,0,1)\rangle &\ \ \  Z = \langle (1,1,1,0)\rangle.
\end{align*}
We have:
\begin{align}
0           &= H(Z|A,B,C)\Comment{$(1,1,1,0)=(1,0,0,0)+(0,1,0,0)+(0,0,1,0)$}\nonumber\\  
            &= H(W|B,C,D)\Comment{$(0,1,1,1)=(0,1,0,0)+(0,0,1,0)+(0,0,0,1)$}\nonumber\\
            &= H(X|A,C,D)\Comment{$(1,0,1,1)=(1,0,0,0)+(0,0,1,0)+(0,0,0,1)$}\nonumber\\
            &= H(Y|W,X,Z)\Comment{$(1,1,0,1)=2^{-1}\cdot((0,1,1,1)+(1,0,1,1)+(1,1,1,0))$}\nonumber\\
            &= H(A|B,D,Y)\Comment{$(1,0,0,0)=(1,1,0,1)-(0,1,0,0)-(0,0,0,1)$}\nonumber\\
            &= H(D|A,W,Z)\Comment{$(0,0,0,1)=(0,1,1,1)+(1,0,0,0)-(1,1,1,0)$}\nonumber\\
            &= H(C|D,Y,Z)\Comment{$(0,0,1,0)=(1,1,1,0)+(0,0,0,1)-(1,1,0,1)$}\nonumber\\
            &= H(B|D,X,Z)\Comment{$(0,1,0,0)=(1,1,1,0)+(0,0,0,1)-(1,0,1,1)$}\nonumber\\
            &= H(C|B,X,Y)\Comment{$(0,0,1,0)=(1,0,1,1)+(0,1,0,0)-(1,1,0,1)$}\nonumber\\
            &= H(C|A,W,Y)\Comment{$(0,0,1,0)=(0,1,1,1)+(1,0,0,0)-(1,1,0,1)$}\nonumber\\
            &= H(B|A,W,X)\Comment{$(0,1,0,0)=(0,1,1,1)+(1,0,0,0)-(1,0,1,1)$}.\nonumber\\
\label{eq:zero-entropies-T8}
\end{align}
Note that the characteristic $3$ assumption is used above
in showing $H(Y|W,X,Z)=0$,
by using the fact that the ranks of $Y$ and $Y\cap \langle W, X, Z\rangle$ are both $1$,
since
$(1,1,0,1) = 2^{-1}\cdot ( (0,1,1,1) + (1,0,1,1) + (1,1,1,0))$,
which holds for scalar fields of characteristic $3$
(in fact, for all characteristics except $2$).

We know $H(A)=H(B)=H(C)=H(D)=H(W)=H(X)=H(Y)=H(Z)=1$.
Also, we have 
$$H(A) + H(B) + H(C) + H(D) = H(A,B,C,D).$$
So, if the inequality in 
Theorem~\ref{thm:T8} were to hold over $V$, 
then we would have
\begin{align*}
1
&= H(A)\\
&\le 8H(Z) + 29H(Y) + 3H(X) + 8H(W) -6H(D) -17H(C) - 8H(B) - 17H(A)\\
&= 8 + 29 + 3 + 8 -6 -17 - 8 - 17\\
&= 0
\end{align*}
which is impossible.
\end{proof}

Consider a network over finite field $F$ with a $(k,n)$ linear code.
The \textit{vector space associated with any message} is defined to be $F^k$.
The \textit{vector space associated with any edge} is defined to be the set of all
possible vectors from $F^n$ that can be carried on that edge
(i.e. taking into account the linear code).

Since each output of a network node is a function of the node's inputs,
the conditional entropy of the vector carried by a node's out-edge,
given the entropies of the vectors carried by the node's in-edges, is zero,
assuming the network messages are uniform random vectors.
The following lemma extends this idea from random variables to vector spaces
and will be useful for the proof of Corollary~\ref{cor:T8-capacity}.

\begin{lem}
Suppose a network has a node with an out-edge (or demand) $x$ 
and in-edges and messages (in some order) $y_1, \dots, y_m$.
Suppose the network has a finite field alphabet and a linear code.
Let us view $X,Y_1, \dots, Y_m$ as the vector spaces associated with 
$x,y_1, \dots, y_m$, respectively.
Then we have $H(X|Y_1, \dots, Y_m)=0$.
\label{lem:conditional-dimensions} 
\end{lem}

\begin{proof}
The vector carried on the node's out-edge (or demand) $x$ is a linear combination
of the vectors carried on the node's in-edges and the node's messages 
$y_1, \dots, y_m$.
Thus, every vector appearing on the node's out-edge (or demand) lies in the span
of the subspaces $Y_1, \dots, Y_m$.
This implies 
$\Dim (X) = \Dim (X \cap \langle Y_1, \dots, Y_m\rangle)$,
or equivalently,
$H(X|Y_1, \dots, Y_m)=0$.
\end{proof}

The following corollary uses the T8 linear rank inequality to derive
capacities and a capacity bound on the T8 network.
Note that although the T8 network itself was used as a guide in obtaining the
T8 linear rank inequality, 
subsequently using the inequality to bound the network capacity is not circular reasoning.

The proof of Corollary~\ref{cor:T8-capacity} below makes use of the T8 linear rank inequality,
and resembles the example shown earlier in \eqref{rv-inequality} for computing the capacity of the Butterfly network
using information inequalities and random variables.

\begin{cor}
For the T8 network,
the linear coding capacity is at most $48/49$ 
over any finite field alphabet of characteristic not equal to $3$.
The linear coding capacity over finite field alphabets of characteristic $3$
and the coding capacity are both equal to $1$.
\label{cor:T8-capacity}
\end{cor}

\begin{proof}
Let $F$ be a finite field alphabet.
Consider a $(k,n)$ linear solution of the T8 network over $F$, 
such that the characteristic of $F$ is not $3$.
Let $A$, $B$, $C$, $D$ be message random variables in the T8 network,
that are uniformly distributed over vectors in $F^k$.
Let $W$, $X$, $Y$, $Z$ be the resulting random variables
associated with the corresponding labeled edges of T8 in Figure~\ref{fig:T8}.

Equations \eqref{eq:zero-entropies-T8} now hold
with random variables $A,B,C,D,W,X,Y,Z$ 
(i.e. not as subspaces as in Theorem~\ref{thm:T8-char3})
by Lemma~\ref{lem:conditional-dimensions}:

\begin{align*}
0           &= H(Z|A,B,C)\Comment{$(n_1,n_2)$}\nonumber\\  
            &= H(W|B,C,D)\Comment{$(n_3,n_4)$}\nonumber\\
            &= H(X|A,C,D)\Comment{$(n_5,n_6)$}\nonumber\\
            &= H(Y|W,X,Z)\Comment{$(n_4,n_7)$}\nonumber\\
            &= H(A|B,D,Y)\Comment{$n_{ 9}$}\nonumber\\
            &= H(D|A,W,Z)\Comment{$n_{10}$}\nonumber\\
            &= H(C|D,Y,Z)\Comment{$n_{11}$}\nonumber\\
            &= H(B|D,X,Z)\Comment{$n_{12}$}\nonumber\\
            &= H(C|B,X,Y)\Comment{$n_{13}$}\nonumber\\
            &= H(C|A,W,Y)\Comment{$n_{14}$}\nonumber\\
            &= H(B|A,W,X)\Comment{$n_{15}$}
\end{align*}
and since the vector spaces $A,B,C,D$ are associated with independent random variables, 
we have 
$$H(A) + H(B) + H(C) + H(D) = H(A,B,C,D)$$
so the T8 inequality in Theorem~\ref{thm:T8} reduces to
\begin{align*}
H(A) &\le  8H(Z) + 29H(Y) + 3H(X) + 8H(W) -6H(D) -17H(C) - 8H(B) - 17H(A).
\end{align*}
Now since
$H(A) {=} H(B) {=} H(C) {=} H(D) {=} k$ and 
$H(W) {=} H(X) {=} H(Y) {=} H(Z) \le n$, 
we have
\begin{align*}
k &\le  8n + 29n + 3n + 8n - 6k -17k - 8k - 17k\\
k/n &\le  48/49.
\end{align*}
So, the linear coding capacity over every characteristic except for 3 is at most $48/49 < 1.$ 

The T8 network has a scalar linear solution over characteristic 3 by using the following
edge functions 
(here we are using the notations $A,B,C,D,W,X,Y,Z$ to denote edge variables rather than vector spaces):
\begin{align*}
Z &=  A + B + C\\
W &=  B + C + D\\
X &=  A + C + D\\
Y &=  W + X + Z.
\end{align*}
and decoding functions:
\begin{align*}
n_{ 9}: A &= (2^{-1}\cdot Y) - B - D\\
n_{10}: D &= W - Z + A\\
n_{11}: C &= Z - (2^{-1}\cdot Y) + D\\
n_{12}: B &= Z - X + D\\
n_{13}: C &= X - (2^{-1}\cdot Y) + B\\
n_{14}: C &= W - (2^{-1}\cdot Y) + A\\
n_{15}: B &= W - X + A
\end{align*}
%
Thus the linear coding capacity for characteristic 3 is at least 1.

We know the coding capacity is at most 1 because every path from
source $A$ to node $n_9$ passes through the single edge $(n_7,n_8)$.
Since the coding capacity is at least as large as the linear coding
capacity for characteristic 3,
we conclude that the coding capacity is exactly equal to 1.
\end{proof}

\newpage
\section{A Linear Rank Inequality for Fields of Characteristic 3}

In the T8 matroid, $W+X+Y+Z=(3,3,3,3)$, which equals $(0,0,0,0)$
in characteristic 3.
We define the \textit{non-T8 matroid} to be the T8 matroid except that
we force the T8's characteristic 3 circuit $\{W,X,Y,Z\}$ to be a base in the non-T8 matroid.
Figure~\ref{fig:nonT8}
is a network
that we call the \textit{non-T8 network},
whose dependencies and independencies are consistent with the non-T8 matroid.
The non-T8 network was designed by the construction process described in \cite{nonshannon}.
Theorem~\ref{thm:nonT8} uses the non-T8 network as a guide to derive a linear rank
inequality valid for characteristic 3.
The new linear rank inequality can then be used to prove the non-T8 network
has linear capacity less than 1 if the field characteristic is 3.

\begin{figure}
\begin{center}
\psfrag{n1}{\Large $n_{1}$}
\psfrag{n2}{\Large $n_{2}$}
\psfrag{n3}{\Large $n_{3}$}
\psfrag{n4}{\Large $n_{4}$}
\psfrag{n5}{\Large $n_{5}$}
\psfrag{n6}{\Large $n_{6}$}
\psfrag{n7}{\Large $n_{7}$}
\psfrag{n8}{\Large $n_{8}$}
\psfrag{n9}{\Large $n_{9}$}
\psfrag{n10}{\Large $n_{10}$}
\psfrag{n11}{\Large $n_{11}$}
\psfrag{n12}{\Large $n_{12}$}
\psfrag{n13}{\Large $n_{13}$}
\psfrag{n14}{\Large $n_{14}$}
\psfrag{n15}{\Large $n_{15}$}
\includegraphics[width=13.5cm]{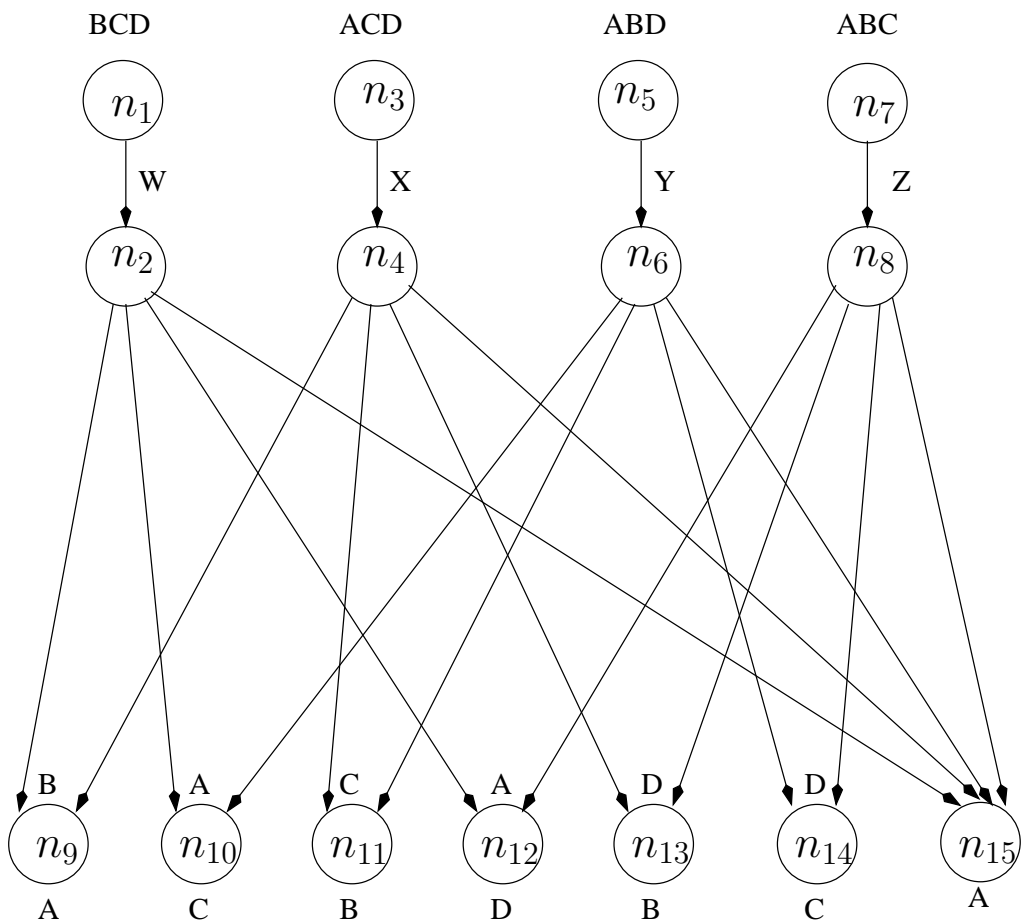}
\end{center}
\caption{ The Non-T8 Network
has source messages $A,B,C,$ and $D$ generated at hidden source nodes with 
certain hidden out-edges pointing to corresponding displayed nodes
$n_1$, $n_3$, $n_5$, $n_7$, and $n_9$--$n_{14}$
(which are labeled by incoming messages above such nodes).
The nodes $n_9$--$n_{15}$ each demand one message, as labeled below such nodes.
}
\label{fig:nonT8}
\end{figure}
\begin{thm}\label{thm:nonT8}
Let $A,B,C,D,W,X,Y$, and $Z$ be subspaces of a vector space $V$ whose scalar field is finite
and of characteristic 3.
Then the following is a linear rank inequality over $V$:
\begin{align*}
H(A) &\le  9H(Z) + 8H(Y) + 5H(X) + 6H(W) - 4H(D) - 12H(C) - 11H(B) -H(A)\\
&\ \ \  +19H(Z|A,B,C) + 17H(Y|A,B,D) + 13H(X|A,C,D) + 11H(W|B,C,D)\\
&\ \ \  +H(A|W,X,Y,Z) + H(A|B,W,X) + 7H(B|D,X,Z) + 4H(B|C,X,Y)\\
&\ \ \  +7H(C|D,Y,Z) + 5H(C|A,W,Y) + 4H(D|A,W,Z)\\
&\ \ \  + 29( H(A) + H(B) + H(C) + H(D) - H(A,B,C,D) ).
\end{align*}
\end{thm}

\begin{proof}
By Lemma~\ref{lemma4} we get linear functions:
\begin{center}
\begin{tabular}{ccc}
     $f_1:W \rightarrow B$, & $f_2:W \rightarrow C$, & $f_3:W \rightarrow D$,  \\
     $f_4:X \rightarrow A$, & $f_5:X \rightarrow C$, & $f_6:X \rightarrow D$,  \\
     $f_7:Y \rightarrow A$, & $f_8:Y \rightarrow B$, & $f_9:Y \rightarrow D$,  \\
     $f_{10}:Z \rightarrow A$, & $f_{11}:Z \rightarrow B$, & $f_{12}:Z \rightarrow C$,  \\
     $f_{13}:A \rightarrow B$, & $f_{14}:A \rightarrow W$, & $f_{15}:A \rightarrow X$,  \\
     $f_{16}:C \rightarrow A$, & $f_{17}:C \rightarrow W$, & $f_{18}:C \rightarrow Y$,  \\
     $f_{19}:B \rightarrow C$, & $f_{20}:B \rightarrow X$, & $f_{21}:B \rightarrow Y$, \\
     $f_{22}:D \rightarrow W$, & $f_{23}:D \rightarrow A$, & $f_{24}:D \rightarrow Z$,  \\
     $f_{25}:B \rightarrow X$, & $f_{26}:B \rightarrow D$, & $f_{27}:B \rightarrow Z$,  \\
     $f_{28}:C \rightarrow Y$, & $f_{29}:C \rightarrow Z$, & $f_{30}:C \rightarrow D$,
\end{tabular}
\begin{tabular}{cccc}
     $f_{31}:A \rightarrow W$, & $f_{32}:A \rightarrow X$, & $f_{33}:A \rightarrow Y$, & $f_{34}:A \rightarrow Z$
\end{tabular}
\end{center}
such that
\begin{align}
    f_1 + f_2 + f_3 &=  I \mbox{ on a subspace of $W$ of codimension $H(W|B,C,D)$}\label{p3:eq:1}\\
    f_4 + f_5 + f_6 &=  I \mbox{ on a subspace of $X$ of codimension $H(X|A,C,D)$}\label{p3:eq:2}\\
    f_7 + f_8 + f_9 &=  I \mbox{ on a subspace of $Y$ of codimension $H(Y|A,B,D)$}\label{p3:eq:3}\\
    f_{10} + f_{11} + f_{12} &=  I \mbox{ on a subspace of $Z$ of codimension $H(Z|A,B,C)$}\label{p3:eq:4}\\
    f_{13} + f_{14} + f_{15} &=  I \mbox{ on a subspace of $A$ of codimension $H(A|B,W,X)$}\label{p3:eq:5}\\
    f_{16} + f_{17} + f_{18} &=  I \mbox{ on a subspace of $C$ of codimension $H(C|A,W,Y)$}\label{p3:eq:6}\\
    f_{19} + f_{20} + f_{21} &=  I \mbox{ on a subspace of $B$ of codimension $H(B|C,X,Y)$}\label{p3:eq:7}\\
    f_{22} + f_{23} + f_{24} &=  I \mbox{ on a subspace of $D$ of codimension $H(D|A,W,Z)$}\label{p3:eq:8}\\
    f_{25} + f_{26} + f_{27} &=  I \mbox{ on a subspace of $B$ of codimension $H(B|D,X,Z)$}\label{p3:eq:9}\\
    f_{28} + f_{29} + f_{30} &=  I \mbox{ on a subspace of $C$ of codimension $H(C|D,Y,Z)$}\label{p3:eq:10}\\
    f_{31} + f_{32} + f_{33} + f_{34} &=  I \mbox{ on a subspace of $A$ of codimension $H(A|W,X,Y,Z)$}.\label{p3:eq:11}
\end{align}
Now combining some functions we obtained from Lemma \ref{lemma4} gives four new functions:
\begin{align*}
    f_{4} f_{32} + f_{7} f_{33} + f_{10} f_{34}  &: A
\rightarrow A \\
    f_{1} f_{31} + f_{8} f_{33} + f_{11} f_{34}  &: A
\rightarrow B \\
    f_{2} f_{31} + f_{5} f_{32} + f_{12} f_{34}  &: A
\rightarrow C \\
    f_{3} f_{31} + f_{6} f_{32} + f_{9} f_{33}  &: A
\rightarrow D.
\end{align*}
Using (\ref{p3:eq:1})--(\ref{p3:eq:4}), (\ref{p3:eq:11}), Lemma~\ref{lemma1},
and Lemma~\ref{lemma3} we know the sum of these four functions is equal to $I$ on a
subspace of $A$ of codimension at most 
$H(W|B,C,D) + H(X|A,C,D) + H(Y|A,B,D) + H(Z|A,B,C) + H(A|W,X,Y,Z)$.

Now applying Lemma~\ref{lemma6} and Lemma~\ref{lemma1} to the functions
$f_{4}\Compose f_{32} + f_{7}\Compose f_{33} + f_{10}\Compose f_{34} - I$, 
$f_{1}\Compose f_{31} + f_{8}\Compose f_{33} + f_{11}\Compose f_{34}$, 
$f_{2}\Compose f_{31} + f_{5}\Compose f_{32} + f_{12}\Compose f_{34}$, 
and 
$f_{3}\Compose f_{31} + f_{6}\Compose f_{32} + f_{9}\Compose f_{33}$,
we get a subspace $\widehat{A}$ of $A$ of codimension at most
\begin{align*}
    \Delta_{\widehat{A}} &=  H(W|B,C,D) + H(X|A,C,D) + H(Y|A,B,D) + H(Z|A,B,C) + H(A|W,X,Y,Z)\\
    & + H(A) + H(B) + H(C) + H(D) - H(A,B,C,D)
\end{align*}
on which
\begin{align}
    f_{4}\Compose f_{32} + f_{7}\Compose f_{33} + f_{10}\Compose f_{34}     &=  I \label{p3:eq:Ahat1}\\
    f_{1}\Compose f_{31} + f_{8}\Compose f_{33} + f_{11}\Compose f_{34}     &=  0 \label{p3:eq:Ahat2}\\
    f_{2}\Compose f_{31} + f_{5}\Compose f_{32} + f_{12}\Compose f_{34}     &=  0 \label{p3:eq:Ahat3}\\
    f_{3}\Compose f_{31} + f_{6}\Compose f_{32} + f_{9}\Compose f_{33}      &=  0. \label{p3:eq:Ahat4}
\end{align}
Similarly, we get a subspace $\overline{A}$ of $A$ of codimension at most
\begin{align*}
    \Delta_{\overline{A}} &=  H(W|B,C,D) + H(X|A,C,D) + H(A|B,W,X)\\
    & + H(A) + H(B) + H(C) + H(D) - H(A,B,C,D)
\end{align*}
on which
\begin{align}
    f_4\Compose f_{15}                               &=  I \label{p3:eq:Abar1}\\
    f_{13} + f_1\Compose f_{14}                      &=  0 \label{p3:eq:Abar2}\\
    f_2\Compose f_{14} + f_5\Compose f_{15}             &=  0 \label{p3:eq:Abar3}\\
    f_3\Compose f_{14} + f_6\Compose f_{15}             &=  0. \label{p3:eq:Abar4}
\end{align}
We get a subspace $\overline{B}$ of $B$ of codimension at most
\begin{align*}
    \Delta_{\overline{B}} &=   H(X|A,C,D) + H(Y|A,B,D)+ H(B|C,X,Y)\\
    & + H(A) + H(B) + H(C) + H(D) - H(A,B,C,D)
\end{align*}
on which
\begin{align}
    f_4\Compose f_{20} + f_7\Compose f_{21}              &=  0 \label{p3:eq:Bbar1}\\
    f_8\Compose f_{21}                                &=  I \label{p3:eq:Bbar2}\\
    f_{19} + f_5\Compose f_{20}                       &=  0 \label{p3:eq:Bbar3}\\
    f_6\Compose f_{20} + f_9\Compose f_{21}              &=  0. \label{p3:eq:Bbar4}
\end{align}
We get a subspace $\widehat{B}$ of $B$ of codimension at most
\begin{align*}
    \Delta_{\widehat{B}} &=  H(X|A,C,D) + H(Z|A,B,C) + H(B|D,X,Z)\\
    & + H(A) + H(B) + H(C) + H(D) - H(A,B,C,D)
\end{align*}
on which
\begin{align}
    f_4\Compose f_{25} + f_{10}\Compose f_{27}                &=  0 \label{p3:eq:Bhat1}\\
    f_{11}\Compose f_{27}                                  &=  I \label{p3:eq:Bhat2}\\
    f_{5}\Compose f_{25} + f_{12}\Compose f_{27}              &=  0 \label{p3:eq:Bhat3}\\
    f_{6}\Compose f_{25} + f_{26}                          &=  0. \label{p3:eq:Bhat4}
\end{align}
We get a subspace $\overline{C}$ of $C$ of codimension at most
\begin{align*}
    \Delta_{\overline{C}} &=  H(W|B,C,D) + H(Y|A,B,D) + H(C|A,W,Y)\\
    & + H(A) + H(B) + H(C) + H(D) - H(A,B,C,D)
\end{align*}
on which
\begin{align}
    f_{16} + f_7\Compose f_{18}                      &=  0 \label{p3:eq:Cbar1}\\
    f_{1}\Compose f_{17} + f_{8}\Compose f_{18}         &=  0 \label{p3:eq:Cbar2}\\
    f_2\Compose f_{17}                               &=  I \label{p3:eq:Cbar3}\\
    f_3\Compose f_{17} + f_9\Compose f_{18}             &=  0. \label{p3:eq:Cbar4}
\end{align}
We get a subspace $\widehat{C}$ of $C$ of codimension at most
\begin{align*}
    \Delta_{\widehat{C}} &=  H(Y|A,B,D) + H(Z|A,B,C) + H(C|D,Y,Z)\\
    & + H(A) + H(B) + H(C) + H(D) - H(A,B,C,D)
\end{align*}
on which
\begin{align}
    f_7\Compose f_{28} + f_{10}\Compose f_{29}           &=  0 \label{p3:eq:Chat1}\\
    f_{8}\Compose f_{28} + f_{11}\Compose f_{29}         &=  0 \label{p3:eq:Chat2}\\
    f_{12}\Compose f_{29}                             &=  I \label{p3:eq:Chat3}\\
    f_9\Compose f_{28} + f_{30}                       &=  0. \label{p3:eq:Chat4}
\end{align}
We get a subspace $\overline{D}$ of $D$ of codimension at most
\begin{align*}
    \Delta_{\overline{D}} &=  H(W|B,C,D) + H(Z|A,B,C) + H(D|A,W,Z)\\
    & + H(A) + H(B) + H(C) + H(D) - H(A,B,C,D)
\end{align*}
on which
\begin{align}
    f_{23} + f_{10}\Compose f_{24}                   &=  0 \label{p3:eq:Dbar1}\\
    f_{1}\Compose f_{22} + f_{11}\Compose f_{24}        &=  0 \label{p3:eq:Dbar2}\\
    f_2\Compose f_{22} + f_{12}\Compose f_{24}          &=  0 \label{p3:eq:Dbar3}\\
    f_3\Compose f_{22}                               &=  I. \label{p3:eq:Dbar4}
\end{align}

Let $\widehat{B}^\ast = f_{11}(f_{27}\widehat{B} \cap f_{29}\widehat{C}) \subseteq \widehat{B}$.
Considering (\ref{p3:eq:Bhat2}) and
(\ref{p3:eq:Chat3}), we can apply Lemma~\ref{lem Inj} to show that
$f_{12}f_{27}$ is injective on $\widehat{B}^\ast$.
By (\ref{p3:eq:Bhat3}), we
know
\begin{align}
& f_{5}f_{25}\mbox{ is injective on $\widehat{B}^\ast$.}\label{p3:f5f25}
\end{align}

Let $\widehat{C}^\ast = f_{12}( f_{29}\widehat{C} \cap f_{27}\widehat{B}) \subseteq \widehat{C}$.
Considering again (\ref{p3:eq:Bhat2}) and
(\ref{p3:eq:Chat3}), we can apply Lemma~\ref{lem Inj} to show that
$f_{11}f_{29}$ is injective on $\widehat{C}^\ast$.
By (\ref{p3:eq:Chat2}), we know
\begin{align}
& f_{8}f_{28}\mbox{ is injective on $\widehat{C}^\ast$.}\label{p3:f8f28}
\end{align}

Let $\overline{A}^\ast = f_4(f_{15}\overline{A} \cap f_{25}\widehat{B}^\ast) \subseteq \overline{A}$.
Considering (\ref{p3:eq:Abar1}) and (\ref{p3:f5f25}),
we can apply Lemma~\ref{lem Inj} to show that $f_5f_{15}$ is injective on
$\overline{A}^\ast$.
By (\ref{p3:eq:Abar3}), we know $f_2f_{14}$ is injective
on $\overline{A}^\ast$ which implies
\begin{align}
&\mbox{$f_{14}$ is injective on $\overline{A}^\ast$.}\label{p3:f14}
\end{align}

Let $\overline{C}^\ast = f_2(f_{17}\overline{C} \cap f_{22}\overline{D}) \subseteq \overline{C}$.
Considering (\ref{p3:eq:Cbar3}) and
(\ref{p3:eq:Dbar4}),we can apply Lemma~\ref{lem Inj} to show that
$f_3f_{17}$ is injective on $\overline{C}^\ast$.
Then by (\ref{p3:eq:Cbar4}),
we know
\begin{align}
& \mbox{ $f_9f_{18}$ is injective on $\overline{C}^\ast$.}\label{p3:f9f18}
\end{align}

Let $\overline{B}^\ast = f_8(f_{21}\overline{B} \cap f_{18}\overline{C}^\ast) \subseteq \overline{B}$.
Considering (\ref{p3:eq:Bbar2}) and \eqref{p3:f9f18}, 
we can apply Lemma~\ref{lem  Inj} to show that
\begin{align}
& f_9f_{21}\mbox{ is injective on $\overline{B}^\ast$.}\label{p3:f9f21}
\end{align}
By (\ref{p3:eq:Bbar4}), we know
\begin{align}
&\mbox{$f_6f_{20}$ is injective on $\overline{B}^\ast$}\label{p3:f6f20}
\end{align}
which implies
\begin{align}
&\mbox{$f_{20}$ is injective on $\overline{B}^\ast$.}\label{p3:f20}
\end{align}

Let us define the functions
\begin{align*}
 g_{14} &= (f_{14} | \overline{A}^\ast)^{-1}\\
 g_{20} &= (f_{20} | \overline{B}^\ast)^{-1}
\end{align*}
where 
$f_{14} | \overline{A}^\ast$
and
$f_{20} | \overline{B}^\ast$
are the restrictions of the functions $f_{14}$ and $f_{20}$
to the sets 
$\overline{A}^\ast$
and
$\overline{B}^\ast$,
respectively.
Now, considering 
(\ref{p3:eq:Bbar2}), (\ref{p3:eq:Bhat2}), (\ref{p3:eq:Cbar3}), and (\ref{p3:eq:Chat3}) we have

\begin{align}
f_1    &= -f_8f_{18}f_2 \mbox{ on $f_{17}\overline{C}$} & \Comment{(\ref{p3:eq:Cbar2})} \label{p3:f1}\\
f_2    &= -f_5f_{15}g_{14} \mbox{ on $f_{14}\overline{A}^\ast$} & \Comment{(\ref{p3:eq:Abar3})} \label{p3:f2}\\
f_3    &= -f_6f_{15}g_{14} \mbox{ on $f_{14}\overline{A}^\ast$ and } f_3 = -f_9f_{18}f_2 \mbox{ on $f_{17}\overline{C}$ } 
           & \Comment{(\ref{p3:eq:Abar4}), (\ref{p3:eq:Cbar4})} \label{p3:f3}\\
f_4    &= -f_7f_{21}g_{20} \mbox{ on $f_{20}\overline{B}^\ast$} & \Comment{(\ref{p3:eq:Bbar1})} \label{p3:f4}\\
f_6    &= -f_9f_{21}g_{20} \mbox{ on $f_{20}\overline{B}^\ast$} & \Comment{(\ref{p3:eq:Bbar4})} \label{p3:f6}\\
f_7    &= -f_4f_{20}f_8 \mbox{ on $f_{21}\overline{B}$} & \Comment{(\ref{p3:eq:Bbar1})} \label{p3:f7}\\
f_9    &= -f_6f_{20}f_8 \mbox{ on $f_{21}\overline{B}$} & \Comment{(\ref{p3:eq:Bbar4})} \label{p3:f9}\\
f_{10} &= -f_4f_{25}f_{11} \mbox{ on $f_{27}\widehat{B}$ and }f_{10} = -f_7f_{28}f_{12} \mbox{ on $f_{29}\widehat{C}$} 
           & \Comment{(\ref{p3:eq:Bhat1}), (\ref{p3:eq:Chat1})} \label{p3:f10}\\
f_{11} &= -f_8f_{28}f_{12} \mbox{ on $f_{29}\widehat{C}$} & \Comment{(\ref{p3:eq:Chat2})} \label{p3:f11}\\
f_{12} &= -f_5f_{25}f_{11} \mbox{ on $f_{27}\widehat{B}$}. & \Comment{(\ref{p3:eq:Bhat3})} \label{p3:f12}
\end{align}

Next, we provide upper bounds for the codimensions of 
$\overline{A}^\ast$, 
$\widehat{B}^\ast$, 
$\overline{B}^\ast$, 
$\widehat{C}^\ast$,
and 
$\overline{C}^\ast$.
From
(\ref{p3:eq:Bhat2}), we know $f_{11}$ is injective on $f_{27}\widehat{B}$ and
$f_{27}$ is injective on $\widehat{B}$.
These facts will be used to arrive on
lines (\ref{p3:Bhatstar:1}) and (\ref{p3:Bhatstar:2}).
From
(\ref{p3:eq:Chat3}), we know $f_{29}$ is injective on $\widehat{C}$, which will
also be used to arrive on line (\ref{p3:Bhatstar:2}).
Lemma~\ref{lemma1}
will be used to arrive on (\ref{p3:Bhatstar:3}).
\begin{align}
\Codim_B\widehat{B}^\ast &=  H(B) - \Dim(\widehat{B}^\ast)\nonumber\\
&=  H(B) - \Dim(f_{11}(f_{27}\widehat{B} \cap f_{29}\widehat{C}))\nonumber\\
&=  H(B) - \Dim(f_{27}\widehat{B} \cap f_{29}\widehat{C})\label{p3:Bhatstar:1}\\
&=  H(B) - H(Z) + \Codim_Z(f_{27}\widehat{B} \cap f_{29}\widehat{C})\nonumber\\
&\le  H(B) - H(Z) + \Codim_Z(f_{27}\widehat{B}) + \Codim_Z( f_{29}\widehat{C})\label{p3:Bhatstar:3}\\
&=  H(B) - H(Z) + H(Z) - \Dim(f_{27}\widehat{B}) + H(Z) - \Dim( f_{29}\widehat{C})\nonumber\\
&=  H(B) + H(Z) - \Dim(\widehat{B}) - \Dim(\widehat{C})\label{p3:Bhatstar:2}\\
&\le  H(B) + H(Z) - H(B) + \Delta_{\widehat{B}} - H(C) + \Delta_{\widehat{C}}\\
&\le  H(Z) - H(C) + \Delta_{\widehat{B}} + \Delta_{\widehat{C}}\\
&\triangleq \Delta_{\widehat{B}^\ast}. \nonumber
\end{align}

From (\ref{p3:eq:Abar1}), we know $f_4$ is injective on $f_{15}\overline{A}$
and $f_{15}$ is injective on $\overline{A}$.
These facts will be used on lines
(\ref{p3:Abarstar:1}) and (\ref{p3:Abarstar:2}).
From (\ref{p3:f5f25}), we know
$f_{25}$ is injective on $\widehat{B}^\ast$, which will also be used to arrive
on line (\ref{p3:Abarstar:2}).
Lemma~\ref{lemma1} will be used to arrive on
(\ref{p3:Abarstar:3}).
\begin{align}
\Codim_A\overline{A}^\ast &=  H(A) - \Dim(\overline{A}^\ast)\nonumber\\
&=  H(A) - \Dim( f_4(f_{25}\widehat{B}^\ast \cap f_{15}\overline{A}) )\nonumber\\
&=  H(A) - \Dim( f_{25}\widehat{B}^\ast \cap f_{15}\overline{A} )\label{p3:Abarstar:1}\\
&=  H(A) - H(X) +  \Codim_X( f_{25}\widehat{B}^\ast \cap f_{15}\overline{A} )\nonumber\\
&\le  H(A) - H(X) +  \Codim_X( f_{25}\widehat{B}^\ast) + \Codim_X( f_{15}\overline{A} )\label{p3:Abarstar:3}\\
&=  H(A) +  H(X) - \Dim( f_{25}\widehat{B}^\ast) - \Dim( f_{15}\overline{A} )\nonumber\\
&=  H(A) +  H(X) - \Dim( \widehat{B}^\ast) - \Dim( \overline{A} )\label{p3:Abarstar:2}\\
&\le  H(A) +  H(X) - H(B) + \Delta_{\widehat{B}^\ast} - H(A) + \Delta_{\overline{A}} \nonumber\\
&=  H(X) - H(B) + H(Z) - H(C) + \Delta_{\widehat{B}} + \Delta_{\widehat{C}} + \Delta_{\overline{A}} \nonumber\\
&\triangleq \Delta_{\overline{A}^\ast}. \nonumber
\end{align}

From (\ref{p3:eq:Cbar3}), we know $f_2$ is injective on $f_{17}\overline{C}$
and $f_{17}$ is injective on $\overline{C}$.
These facts will be used to arrive
on lines (\ref{p3:Cbarstar:1}) and (\ref{p3:Cbarstar:2}).
From
(\ref{p3:eq:Dbar4}), we know $f_{22}$ is injective on $\overline{D}$, which
will also be used on line (\ref{p3:Cbarstar:2}).
Lemma~\ref{lemma1} will be
used to arrive on (\ref{p3:Cbarstar:3}).
\begin{align}
\Codim_C\overline{C}^\ast &=  H(C) - \Dim(\overline{C}^\ast)\nonumber\\
&=  H(C) - \Dim( f_2(f_{17}\overline{C} \cap f_{22}\overline{D}) )\nonumber\\
&=  H(C) - \Dim( f_{17}\overline{C} \cap f_{22}\overline{D} )\label{p3:Cbarstar:1}\\
&=  H(C) - H(W) + \Codim_W( f_{17}\overline{C} \cap f_{22}\overline{D} )\nonumber\\
&\le  H(C) - H(W) + \Codim_W( f_{17}\overline{C}) + \Codim_W( f_{22}\overline{D} )\label{p3:Cbarstar:3}\\
&=  H(C) - H(W) + H(W) - \Dim( f_{17}\overline{C}) + H(W) - \Dim( f_{22}\overline{D} )\nonumber\\
&=  H(C) + H(W) - \Dim( \overline{C}) - \Dim( \overline{D} )\label{p3:Cbarstar:2}\\
&\le  H(C) + H(W) - H(C) + \Delta_{\overline{C}} - H(D) + \Delta_{\overline{D}} \nonumber\\
&=  H(W) - H(D) + \Delta_{\overline{C}} + \Delta_{\overline{D}} \nonumber\\
&\triangleq \Delta_{\overline{C}^\ast}. \nonumber
\end{align}

From (\ref{p3:eq:Bbar2}), we know $f_8$ is injective on $f_{21}\overline{B}$
and $f_{21}$ is injective on $\overline{B}$.
These facts will be used to arrive
on lines (\ref{p3:Bbarstar:1}) and (\ref{p3:Bbarstar:2}).
From
(\ref{p3:f9f18}), we know $f_{18}$ is injective on $\overline{C}^\ast$, which
will also be used on line (\ref{p3:Bbarstar:2}).
Lemma~\ref{lemma1} will be
used to arrive on (\ref{p3:Bbarstar:3}).
\begin{align}
\Codim_B\overline{B}^\ast &=  H(B) - \Dim(\overline{B}^\ast)\nonumber\\
&=  H(B) - \Dim( f_8(f_{21}\overline{B} \cap f_{18}\overline{C}^\ast) )\nonumber\\
&=  H(B) - \Dim( f_{21}\overline{B} \cap f_{18}\overline{C}^\ast )\label{p3:Bbarstar:1}\\
&=  H(B) - H(Y) + \Codim_Y( f_{21}\overline{B} \cap f_{18}\overline{C}^\ast )\nonumber\\
&\le  H(B) - H(Y) + \Codim_Y( f_{21}\overline{B}) + \Codim_Y( f_{18}\overline{C}^\ast )\label{p3:Bbarstar:3}\\
&=  H(B) - H(Y) + H(Y) - \Dim( f_{21}\overline{B}) + H(Y) - \Dim( f_{18}\overline{C}^\ast )\nonumber\\
&=  H(B) + H(Y) - \Dim( \overline{B}) - \Dim( \overline{C}^\ast )\label{p3:Bbarstar:2}\\
&\le  H(B) + H(Y) - H(B) + \Delta_{\overline{B}} - H(C) + \Delta_{\overline{C}^\ast} \nonumber\\
&=  H(Y) - H(C) + \Delta_{\overline{B}} + \Delta_{\overline{C}^\ast} \nonumber\\
&=  H(Y) - H(C) + H(W) - H(D) + \Delta_{\overline{C}} + \Delta_{\overline{D}} + \Delta_{\overline{B}} \nonumber\\
&\triangleq \Delta_{\overline{B}^\ast}. \nonumber
\end{align}

From (\ref{p3:eq:Chat3}), we know $f_{12}$ is injective on $f_{29}\widehat{C}$
and $f_{29}$ is injective on $\widehat{C}$.
These facts will be used to arrive
on lines (\ref{p3:Chatstar:1}) and (\ref{p3:Chatstar:2}).
From
(\ref{p3:eq:Bhat2}), we know $f_{27}$ is injective on $\widehat{B}$, which will
also be used on line (\ref{p3:Chatstar:2}).
Lemma~\ref{lemma1} will be used
to arrive on (\ref{p3:Chatstar:3}).
\begin{align}
\Codim_C\widehat{C}^\ast &=  H(C) - \Dim(\widehat{C}^\ast)\nonumber\\
&=  H(C) - \Dim( f_{12}(f_{27}\widehat{B} \cap f_{29}\widehat{C}) )\nonumber\\
&=  H(C) - \Dim( f_{27}\widehat{B} \cap f_{29}\widehat{C} )\label{p3:Chatstar:1}\\
&=  H(C) - H(Z) + \Codim_Z( f_{27}\widehat{B} \cap f_{29}\widehat{C} )\nonumber\\
&\le  H(C) - H(Z) + \Codim_Z( f_{27}\widehat{B}) + \Codim_Z( f_{29}\widehat{C} )\label{p3:Chatstar:3}\\
&=  H(C) - H(Z) + H(Z) - \Dim( f_{27}\widehat{B}) + H(Z) - \Dim( f_{29}\widehat{C} )\nonumber\\
&=  H(C) + H(Z) - \Dim( \widehat{B}) - \Dim( \widehat{C} )\label{p3:Chatstar:2}\\
&\le  H(C) + H(Z) - H(B) + \Delta_{\widehat{B}} - H(C) + \Delta_{\widehat{C}}\nonumber\\
&=  H(Z) - H(B) + \Delta_{\widehat{B}} + \Delta_{\widehat{C}}\nonumber\\
&\triangleq \Delta_{\widehat{C}^\ast}. \nonumber
\end{align}

Let $t \in A$.
Now, we will assume $t$ satisfies conditions (\ref{p3:ass1})--(\ref{p3:ass7}).
The justifications can be found below.
\begin{align}
& t \in \widehat{A}\mbox{   ; this is true on a subspace of $A$ of codimension at most } \Delta_{\widehat{A}}\label{p3:ass1}\\
& f_{32}t \in f_{20}\overline{B}^\ast \cap f_{25}\widehat{B}^\ast \mbox{   ; this is true on a subspace of $A$ of codimension at most }\nonumber\\
&\qquad\qquad\qquad\qquad\qquad\qquad\qquad\qquad 2H(X) - 2H(B) + \Delta_{\overline{B}^\ast} + \Delta_{\widehat{B}^\ast}\label{p3:ass3}\\
& f_{33}t \in f_{28}\widehat{C}^\ast \cap f_{21}\overline{B}^\ast \mbox{    ; this is true on a subspace of $A$ of codimension at most }\nonumber\\
&\qquad\qquad\qquad\qquad\qquad\qquad\qquad\qquad 2H(Y) - H(B) - H(C) + \Delta_{\overline{B}^\ast} + \Delta_{\widehat{C}^\ast}\label{p3:ass4}\\
& f_{34}t \in f_{29}\widehat{C}^\ast \cap f_{27}\widehat{B}^\ast \mbox{   ; this is true on a subspace of $A$ of codimension at most } \nonumber\\
&\qquad\qquad\qquad\qquad\qquad\qquad\qquad\qquad 2H(Z) - H(C) - H(B) + \Delta_{\widehat{C}^\ast} + \Delta_{\widehat{B}^\ast}\label{p3:ass5}\\
& f_{18}f_2f_{31}t \in f_{21}\overline{B}^\ast \cap f_{28}\widehat{C}^\ast \mbox{   ; this is true on a subspace of $A$ of codimension at most } \nonumber\\
&\qquad\qquad\qquad\qquad\qquad\qquad\qquad\qquad 2H(Y) - H(B) - H(C) + \Delta_{\overline{B}^\ast} + \Delta_{\widehat{C}^\ast}\label{p3:ass6}
\end{align}
Now, we need to make two assumptions on $t$ simultaneously.
\begin{align}
& f_{31}t \in f_{17}\overline{C} \cap f_{14}\overline{A}^\ast 
  \mbox{ and }f_{15}g_{14}f_{31}t \in f_{20}\overline{B}^\ast \cap f_{25}\widehat{B}^\ast ; \nonumber\\
&  \mbox{this is true on a subspace of $A$ of codimension at most}\nonumber\\
&\qquad\qquad 2H(X) - 2H(B) + 2H(W) - H(C) - H(A) + \Delta_{\overline{C}} + \Delta_{\overline{A}^\ast} + \Delta_{\overline{B}^\ast} + \Delta_{\widehat{B}^\ast}\label{p3:ass7}
\end{align}

To justify (\ref{p3:ass3}), first we know $f_{20}$ is injective on
$\overline{B}^\ast$ by (\ref{p3:f6f20}).
Then by Lemma~\ref{lemma3}, we
know $f_{32}t \in f_{20}\overline{B}^\ast$ on a subspace of $A$ of codimension
at most $H(X) - H(B) + \Codim_B(\overline{B}^\ast) \le H(X) - H(B) + \Delta_{\overline{B}^\ast}$.
By (\ref{p3:f5f25}), we also know $f_{25}$ is
injective on $\widehat{B}^\ast$.
Then by Lemma~\ref{lemma3}, we know
$f_{32}t \in f_{25}\widehat{B}^\ast$ on a subspace of $A$ of codimension at
most $H(X) - H(B) + \Codim_B(\widehat{B}^\ast) \le H(X) - H(B) + \Delta_{\widehat{B}^\ast}$.
Then using Lemma~\ref{lemma1}, we have $f_{32}t \in f_{20}\overline{B}^\ast \cap f_{25}\widehat{B}^\ast$ on a subspace of $A$
of codimension at most $2H(X) - 2H(B) \Delta_{\overline{B}^\ast} + \Delta_{\widehat{B}^\ast}$.
Conditions (\ref{p3:ass4})--(\ref{p3:ass6}) can be justified similarly.

To justify (\ref{p3:ass7}), first we know $f_{17}$ is injective on
$\overline{C}$ by (\ref{p3:eq:Cbar3}).
Then by Lemma~\ref{lemma3}, we know
$f_{31}t \in f_{17}\overline{C}$ on a subspace of $A$ of codimension at most
$H(W) - H(C) + \Codim_C(\overline{C}) \le H(W) - H(C) + \Delta_{\overline{C}}$.
By (\ref{p3:f14}), we also know $f_{14}$ is injective
on $\overline{A}^\ast$.
Then by Lemma~\ref{lemma3}, we know $f_{31}t \in f_{14}\overline{A}^\ast$ on a subspace of $A$ of codimension at most $H(W) -
H(A) + \Codim_A(\overline{A}^\ast) \le H(W) - H(A) + \Delta_{\overline{A}^\ast}$.
Then using Lemma~\ref{lemma1}, we have
\begin{align*}
&\mbox{$f_{31}t \in f_{17}\overline{C} \cap f_{14}\overline{A}^\ast$}
\end{align*}
on a subspace, $S$, of $A$ of codimension at most $2H(W) - H(C) - H(A) + \Delta_{\overline{C}} + \Delta_{\overline{A}^\ast}$.
Since $f_{14}$ is
injective on $\overline{A}^\ast$, the function $f_{15}g_{14}f_{31}$ is
defined on $S$.
Using the same technique as before we can show that
\begin{align*}
f_{15}g_{14}f_{31}t \in f_{20}\overline{B}^\ast \cap f_{25}\widehat{B}^\ast
\end{align*}
on a subspace, $\overline{S}$, of codimension with respect to $S$ at most
$2H(X) - 2H(B) + \Delta_{\overline{B}^\ast} + \Delta_{\widehat{B}^\ast}$.
Thus
both conditions are true on $\overline{S}$, which has codimension with respect
to $A$ at most 
$\Codim_{S}\overline{S} + \Codim_A{S} \le 2H(X) - 2H(B) + 2H(W) - H(C) - H(A) 
  + \Delta_{\overline{C}} + \Delta_{\overline{A}^\ast} + \Delta_{\overline{B}^\ast} + \Delta_{\widehat{B}^\ast}$.

Our final goal is to show that $t = 3x$ for some $x$ so that we may conclude
that $t = 0$ if the characteristic is 3.
We will accomplish this by using
(\ref{p3:eq:Ahat1}) and by proving that $f_4f_{32}t = f_7f_{33}t =
f_{10}f_{34}t$.
\begin{claim}\label{p3:claim1}
$f_4f_{32}t = f_{10}f_{34}t$
\end{claim}

\begin{proof}
First we must show that $f_{28}f_{12}f_{34}t = f_{21}g_{20}f_{32}t$.
By (\ref{p3:eq:Ahat2}), we know
\begin{align*}
f_8f_{33}t &=  -f_{11}f_{34}t - f_1f_{31}t.
\end{align*}
Then by using (\ref{p3:f11}) and condition (\ref{p3:ass5}), we have
\begin{align*}
f_8f_{33}t &=  f_8f_{28}f_{12}f_{34}t - f_1f_{31}t.
\end{align*}
Now, by using (\ref{p3:f1}) and condition (\ref{p3:ass7}), we have
\begin{align*}
f_8f_{33}t &=  f_8f_{28}f_{12}f_{34}t + f_8f_{18}f_2f_{31}t.
\end{align*}
By (\ref{p3:f8f28}), we know $f_8$ is injective on $f_{28}\widehat{C}^\ast$.
By condition (\ref{p3:ass4}), we know $f_{33}t \in f_{28}\widehat{C}^\ast$.
By condition (\ref{p3:ass6}), we know $f_{18}f_2f_{31}t \in f_{28}\widehat{C}^\ast$.
By condition (\ref{p3:ass5}), we know $f_{34}t \in f_{29}\widehat{C}^\ast$.
Using (\ref{p3:eq:Chat3}), we know $f_{12}f_{34}t \in \widehat{C}^\ast$.
Thus, we have
\begin{align}
f_{33}t &=  f_{28}f_{12}f_{34}t + f_{18}f_2f_{31}t. \label{p3:claim1:1}
\end{align}
By (\ref{p3:eq:Ahat4}), we have
\begin{align*}
f_9f_{33}t &=  -f_6f_{32}t - f_3f_{31}t.
\end{align*}
Then by using (\ref{p3:f6}) and condition (\ref{p3:ass3}), we have
\begin{align*}
f_9f_{33}t &=  f_9f_{21}g_{20}f_{32}t - f_3f_{31}t.
\end{align*}
Now, by using (\ref{p3:f3}) and condition (\ref{p3:ass7}), we have
\begin{align*}
f_9f_{33}t &=  f_9f_{21}g_{20}f_{32}t + f_9f_{18}f_2f_{31}t.
\end{align*}
By (\ref{p3:f9f21}), we know $f_9$ is injective on $f_{21}\overline{B}^\ast$.
By condition (\ref{p3:ass4}), we know $f_{33}t \in f_{21}\overline{B}^\ast$.
By condition (\ref{p3:ass3}), we know $f_{32}t \in f_{20}\overline{B}^\ast$ so $f_{21}g_{20}f_{32}t \in f_{21}\overline{B}^\ast$.
By condition (\ref{p3:ass6}), we know $f_{18}f_2f_{31}t \in f_{21}\overline{B}^\ast$.
Thus, we have
\begin{align}
f_{33}t &=  f_{21}g_{20}f_{32}t + f_{18}f_2f_{31}t. \label{p3:claim1:2}
\end{align}
Now, setting (\ref{p3:claim1:1}) and (\ref{p3:claim1:2}) equal to each other, we have
\begin{align}
f_{21}g_{20}f_{32}t &=  f_{28}f_{12}f_{34}t. \label{p3:claim1:3}
\end{align}
By (\ref{p3:f4}) and condition (\ref{p3:ass3}), we know
\begin{align*}
f_4f_{32}t &=  -f_7f_{21}g_{20}f_{32}t.
\end{align*}
Using (\ref{p3:claim1:3}), we have
\begin{align*}
f_4f_{32}t &=  -f_7f_{28}f_{12}f_{34}t.
\end{align*}
Then using (\ref{p3:f10}) and condition (\ref{p3:ass5}), we know
\begin{align*}
f_4f_{32}t &=  f_{10}f_{34}t.
\end{align*}
\end{proof}

\begin{claim}\label{p3:claim2}
$f_7f_{33}t = f_{10}f_{34}t$.
\end{claim}

\begin{proof}
First we must show that $f_{25}f_{11}f_{34}t = f_{20}f_8f_{33}t$.
By (\ref{p3:eq:Ahat3}), we know
\begin{align*}
f_5f_{32}t &=  -f_{12}f_{34}t - f_2f_{31}t.
\end{align*}
Then by using (\ref{p3:f12}) and condition (\ref{p3:ass5}), we have
\begin{align*}
f_5f_{32}t &=  f_5f_{25}f_{11}f_{34}t - f_2f_{31}t.
\end{align*}
Now, by using (\ref{p3:f2}) and condition (\ref{p3:ass7}), we have
\begin{align*}
f_5f_{32}t &=  f_5f_{25}f_{11}f_{34}t + f_5f_{15}g_{14}f_{31}t.
\end{align*}
By (\ref{p3:f5f25}), we know $f_5$ is injective on $f_{25}\widehat{B}^\ast$.
By condition (\ref{p3:ass3}), we know $f_{32}t \in f_{25}\widehat{B}^\ast$.
By condition (\ref{p3:ass5}), we know $f_{34}t \in f_{27}\widehat{B}^\ast$.
Now, using (\ref{p3:eq:Bhat2}), we know $f_{11}f_{34}t \in \widehat{B}^\ast$.
By condition (\ref{p3:ass7}), we know $f_{15}g_{14}f_{31}t \in f_{25}\widehat{B}^\ast$.
Thus, we have
\begin{align}
f_{32}t &=  f_{25}f_{11}f_{34}t + f_{15}g_{14}f_{31}t. \label{p3:claim2:1}
\end{align}
By (\ref{p3:eq:Ahat4}), we have
\begin{align*}
f_6f_{32}t &=  -f_9f_{33}t - f_3f_{31}t.
\end{align*}
Then using (\ref{p3:f9}) and condition (\ref{p3:ass4}), we have
\begin{align*}
f_6f_{32}t &=  f_6f_{20}f_{8}f_{33}t - f_3f_{31}t.
\end{align*}
Now, by using (\ref{p3:f3}) and condition (\ref{p3:ass7}), we have
\begin{align*}
f_6f_{32}t &=  f_6f_{20}f_{8}f_{33}t + f_6f_{15}g_{14}f_{31}t.
\end{align*}
By (\ref{p3:f6f20}), we know that $f_{6}$ is injective on $f_{20}\overline{B}^\ast$.
By condition (\ref{p3:ass3}), we know $f_{32}t \in f_{20}\overline{B}^\ast$.
By condition (\ref{p3:ass4}), we know $f_{33}t \in f_{21}\overline{B}^\ast$.
Now, using (\ref{p3:eq:Bbar2}), we know $f_8f_{33}t \in \overline{B}^\ast$.
By condition (\ref{p3:ass7}), we know $f_{15}g_{14}f_{31}t \in f_{20}\overline{B}^\ast$.
Thus, we have
\begin{align}
f_{32}t &=  f_{20}f_{8}f_{33}t + f_{15}g_{14}f_{31}t. \label{p3:claim2:2}
\end{align}
Now, setting (\ref{p3:claim2:1}) and (\ref{p3:claim2:2}) equal to each other, we have
\begin{align}
f_{25}f_{11}f_{34}t &=  f_{20}f_{8}f_{33}t. \label{p3:claim2:3}
\end{align}
By (\ref{p3:f7}) and condition (\ref{p3:ass4}), we know
\begin{align*}
f_7f_{33}t &=  -f_4f_{20}f_{8}f_{33}t.
\end{align*}
Using (\ref{p3:claim2:3}), we have
\begin{align*}
f_7f_{33}t &=  -f_4f_{25}f_{11}f_{34}t.
\end{align*}
Then using (\ref{p3:f10}) and condition (\ref{p3:ass5}), we know
\begin{align*}
f_7f_{33}t &=  f_{10}f_{34}t.
\end{align*}

\end{proof}

Now, by \eqref{p3:eq:Ahat1} and the two claims,
we have
\begin{align*}
t &=  f_4f_{32}t + f_7f_{33}t + f_{10}f_{34}t\\
  &=  f_{10}f_{34}t + f_{10}f_{34}t + f_{10}f_{34}t\\
  &=  3f_{10}f_{34}t.
\end{align*}
Thus if the field has characteristic 3, then
\begin{align}
t &=  0.
\end{align}
No nonzero $t$ can satisfy all of the conditions (\ref{p3:ass1})--(\ref{p3:ass7}), 
so we must have
\begin{align*}
H(A) &\le  \Delta_{\widehat{A}} + 2H(W) - H(C) - H(A) + \Delta_{\overline{C}} + \Delta_{\overline{A}^\ast} \\
&\ \ \  + 2H(X) - 2H(B) + \Delta_{\overline{B}^\ast} + \Delta_{\widehat{B}^\ast}\\
&\ \ \  + 2H(Y) - H(B) - H(C) + \Delta_{\overline{B}^\ast} + \Delta_{\widehat{C}^\ast}\\
&\ \ \  + 2H(Z) - H(C) - H(B) + \Delta_{\widehat{C}^\ast} + \Delta_{\widehat{B}^\ast}\\
&\ \ \  + 2H(Y) - H(B) - H(C) + \Delta_{\overline{B}^\ast} + \Delta_{\widehat{C}^\ast}\\
&\ \ \  + 2H(X) - 2H(B) + \Delta_{\overline{B}^\ast} + \Delta_{\widehat{B}^\ast}\\
&=  2H(Z) + 4H(Y) + 4H(X) + 2H(W) - 4H(C) - 7H(B) - H(A) \\
&\ \ \  + \Delta_{\overline{A}^\ast} + 4\Delta_{\overline{B}^\ast} + 3\Delta_{\widehat{B}^\ast} + 3\Delta_{\widehat{C}^\ast} + \Delta_{\widehat{A}}+ \Delta_{\overline{C}}\\
&=  2H(Z) + 4H(Y) + 4H(X) + 2H(W) - 4H(C) - 7H(B) - H(A) \\
&\ \ \  + H(X) - H(B) + H(Z) - H(C) + \Delta_{\widehat{B}} + \Delta_{\widehat{C}} + \Delta_{\overline{A}}\\
&\ \ \  + 4( H(Y) - H(C) + H(W) - H(D) + \Delta_{\overline{C}} + \Delta_{\overline{D}} + \Delta_{\overline{B}} )\\
&\ \ \  + 3( H(Z) - H(C) + \Delta_{\widehat{B}} + \Delta_{\widehat{C}} )\\
&\ \ \  + 3( H(Z) - H(B) + \Delta_{\widehat{B}} + \Delta_{\widehat{C}} )\\
&\ \ \  + \Delta_{\widehat{A}}+ \Delta_{\overline{C}}\\
&=  9H(Z) + 8H(Y) + 5H(X) + 6H(W) - 4H(D) - 12H(C) - 11H(B) -H(A)\\
&\ \ \  + \Delta_{\widehat{A}} + \Delta_{\overline{A}} + 7\Delta_{\widehat{B}} + 4\Delta_{\overline{B}} + 7\Delta_{\widehat{C}} + 5\Delta_{\overline{C}} + 4\Delta_{\overline{D}} \\
&=  9H(Z) + 8H(Y) + 5H(X) + 6H(W) - 4H(D) - 12H(C) - 11H(B) -H(A)\\
&\ \ \  + H(W|B,C,D) + H(X|A,C,D) + H(Y|A,B,D) + H(Z|A,B,C) \\
&\ \ \  + H(A|W,X,Y,Z)+ H(W|B,C,D) + H(X|A,C,D) + H(A|B,W,X) \\
&\ \ \  + 7( H(X|A,C,D) + H(Z|A,B,C) + H(B|D,X,Z) )\\
&\ \ \  + 4( H(X|A,C,D) + H(Y|A,B,D)+ H(B|C,X,Y) )\\
&\ \ \  + 7( H(Y|A,B,D) + H(Z|A,B,C) + H(C|D,Y,Z) )\\
&\ \ \  + 5( H(W|B,C,D) + H(Y|A,B,D) + H(C|A,W,Y) )\\
&\ \ \  + 4( H(W|B,C,D) + H(Z|A,B,C) + H(D|A,W,Z) )\\
&\ \ \  + 29( H(A) + H(B) + H(C) + H(D) - H(A,B,C,D) )\\
&=  9H(Z) + 8H(Y) + 5H(X) + 6H(W) - 4H(D) - 12H(C) - 11H(B) -H(A)\\
&\ \ \  +19H(Z|A,B,C) + 17H(Y|A,B,D) + 13H(X|A,C,D) + 11H(W|B,C,D)\\
&\ \ \  +H(A|W,X,Y,Z) + H(A|B,W,X) + 7H(B|D,X,Z) + 4H(B|C,X,Y)\\
&\ \ \  +7H(C|D,Y,Z) + 5H(C|A,W,Y) + 4H(D|A,W,Z)\\
&\ \ \  + 29( H(A) + H(B) + H(C) + H(D) - H(A,B,C,D) ).
\end{align*}
\end{proof}

The next theorem demonstates that the inequality in Theorem~\ref{thm:nonT8}
does not in general hold for vector spaces with finite fields of characteristic other than $3$.

\begin{thm}\label{thm:non-T8-non-char3}
For each prime number $p\ne 3$
there exists a vector space $V$ with a finite scalar field of characteristic $p$ 
such that the non-T8 inequality in Theorem~\ref{thm:nonT8} 
is not a linear rank inequality over $V$.
\end{thm}

\begin{proof}
Let $V$ be the vector space of $4$-dimensional vectors whose components are from $GF(p)$,
and define the following subspaces of $V$:
\begin{align*}
    A = \langle (1,0,0,0)\rangle & B = \langle (0,1,0,0)\rangle\\
    C = \langle (0,0,1,0)\rangle & D = \langle (0,0,0,1)\rangle\\
    W = \langle (0,1,1,1)\rangle & X = \langle (1,0,1,1)\rangle\\
    Y = \langle (1,1,0,1)\rangle & Z = \langle (1,1,1,0)\rangle.
\end{align*}

We have:
\begin{align}
0
            &= H(W|B,C,D)    & \Comment{$(0,1,1,1)=(0,1,0,0)+(0,0,1,0)+(0,0,0,1)$}\nonumber\\
            &= H(X|A,C,D)    & \Comment{$(1,0,1,1)=(1,0,0,0)+(0,0,1,0)+(0,0,0,1)$}\nonumber\\
            &= H(Y|A,B,D)    & \Comment{$(1,1,0,1)=(1,0,0,0)+(0,1,0,0)+(0,0,0,1)$}\nonumber\\
            &= H(Z|A,B,C)    & \Comment{$(1,1,1,0)=(1,0,0,0)+(0,1,0,0)+(0,0,1,0)$}\nonumber\\
            &= H(A|B,W,X)    & \Comment{$(1,0,0,0)=(1,0,1,1)+(0,1,0,0)-(0,1,1,1)$}\nonumber\\
            &= H(C|A,W,Y)    & \Comment{$(0,0,1,0)=(0,1,1,1)+(1,0,0,0)-(1,1,0,1)$}\nonumber\\
            &= H(B|C,X,Y)    & \Comment{$(0,1,0,0)=(1,1,0,1)+(0,0,1,0)-(1,0,1,1)$}\nonumber\\
            &= H(D|A,W,Z)    & \Comment{$(0,0,0,1)=(0,1,1,1)+(1,0,0,0)-(1,1,1,0)$}\nonumber\\
            &= H(B|D,X,Z)    & \Comment{$(0,1,0,0)=(1,1,1,0)+(0,0,0,1)-(1,0,1,1)$}\nonumber\\
            &= H(C|D,Y,Z)    & \Comment{$(0,0,1,0)=(1,1,1,0)+(0,0,0,1)-(1,1,0,1)$}\nonumber\\
            &= H(A|W,X,Y,Z)  & \Comment{$(1,0,0,0)=3^{-1}( (1,0,1,1){+}(1,1,0,1){+}(1,1,1,0) {-} 2(0,1,1,1))$}.
\label{eq:zero-entropies-nonT8}
\end{align}
We know $H(A)=H(B)=H(C)=H(D)=H(W)=H(X)=H(Y)=H(Z)=1$,
Also, we have 
$$H(A) + H(B) + H(C) + H(D) = H(A,B,C,D).$$
So, if the inequality in 
Theorem~\ref{thm:nonT8} were to hold over $V$, 
then we would have
\begin{align*}
1 
&= H(A) \\
&\le 9H(Z) + 8H(Y) + 5H(X) + 6H(W) - 4H(D) - 12H(C) - 11H(B) -H(A)\\
&=  9 + 8 + 5 + 6 -4 -12 - 11 - 1\\
&= 0
\end{align*}
which is impossible.
\end{proof}

\begin{cor}
For the non-T8 network,
the linear coding capacity is at most $28/29$ 
over any finite field alphabet of characteristic equal to $3$.
The linear coding capacity over finite field alphabets of characteristic not $3$
and the coding capacity are all equal to $1$.
\label{cor:nonT8-capacity}
\end{cor}

\begin{proof}
Let $F$ be a finite field alphabet.
Consider a $(k,n)$ linear solution of the non-T8 network over $F$, 
such that the characteristic of $F$ is $3$.
Let $A$, $B$, $C$, $D$ be message random variables in the T8 network,
that are uniformly distributed over vectors in $F^k$.
Let $W$, $X$, $Y$, $Z$ be the resulting random variables
associated with the corresponding labeled edges of T8 in Figure~\ref{fig:nonT8}.

Equations \eqref{eq:zero-entropies-nonT8} now hold 
with random variables $A,B,C,D,W,X,Y,Z$ are taken as random variables 
(i.e. not as subspaces as in Theorem~\ref{thm:non-T8-non-char3})
by Lemma~\ref{lem:conditional-dimensions}:
\begin{align*}
0
            &= H(W|B,C,D)    & \Comment{$(n_1,n_2)$}\nonumber\\
            &= H(X|A,C,D)    & \Comment{$(n_3,n_4)$}\nonumber\\
            &= H(Y|A,B,D)    & \Comment{$(n_5,n_6)$}\nonumber\\
            &= H(Z|A,B,C)    & \Comment{$(n_7,n_8)$}\nonumber\\
            &= H(A|B,W,X)    & \Comment{$n_{ 9}$}\nonumber\\
            &= H(C|A,W,Y)    & \Comment{$n_{10}$}\nonumber\\
            &= H(B|C,X,Y)    & \Comment{$n_{11}$}\nonumber\\
            &= H(D|A,W,Z)    & \Comment{$n_{12}$}\nonumber\\
            &= H(B|D,X,Z)    & \Comment{$n_{13}$}\nonumber\\
            &= H(C|D,Y,Z)    & \Comment{$n_{14}$}\nonumber\\
            &= H(A|W,X,Y,Z)  & \Comment{$n_{15}$}
\end{align*}
and since the source message $A,B,C,D$ are independent random variables, 
we have 
$$H(A) + H(B) + H(C) + H(D) = H(A,B,C,D)$$
so the non-T8 inequality in Theorem~\ref{thm:nonT8} reduces to
\begin{align*}
H(A) &\le  9H(Z) + 8H(Y) + 5H(X) + 6H(W) - 4H(D) - 12H(C) - 11H(B) -H(A).
\end{align*}
Now, since 
$H(A) = H(B) = H(C) = H(D) = k$ and 
$H(W) = H(X) = H(Y) = H(Z) \le n$, 
we have
\begin{align*}
k &\le  9n + 8n + 5n + 6n - 4k -12k - 11k - k\\
k/n &\le  28/29.
\end{align*}
So, the linear coding capacity over characteristic 3 is at most $28/29 < 1.$ 

The non-T8 network 
has a scalar linear solution  over every characteristic except for 3 by using the
following edge functions
(here we are using the notations $A,B,C,D,W,X,Y,Z$ to denote edge variables rather than vector spaces):
\begin{align*}
W &=  B + C + D\\
X &=  A + C + D\\
Y &=  A + B + D\\
Z &=  A + B + C
\end{align*}
and decoding functions:
\begin{align*}
n_{ 9}: A &= X - W + B\\
n_{10}: C &= W - Y + A\\
n_{11}: B &= Y - X + C\\
n_{12}: D &= W - Z + A\\
n_{13}: B &= Z - X + D\\
n_{14}: C &= Z - Y + D\\
n_{15}: A &= 3^{-1} \cdot ( X + Y + Z - 2W ).
\end{align*}
We know the coding capacity is at most 1 because there is a unique path from
source $A$ to node $n_9$ (through node $n_4$).
Since the coding capacity is at least as large as the linear coding
capacity for characteristics other than 3,
we conclude that the coding capacity is exactly equal to 1.
\end{proof}


\end{document}